\documentclass[a4paper]{article}

\usepackage[figuresright]{rotating}
\usepackage{multirow}
\usepackage{graphics}
\usepackage{graphicx}
\usepackage{makecell}
\usepackage{color}
\usepackage{graphicx}
\usepackage{amsmath}
\usepackage{amsthm}
\usepackage{amssymb}
\usepackage{subfig}
\usepackage{float}
\usepackage{balance}
\usepackage{cite}
\usepackage{booktabs}
\usepackage{cases}
\usepackage{bm}
\usepackage{mathrsfs}
\usepackage[top=2.5cm, bottom=2.5cm, left=3cm, right=3cm]{geometry}

\makeatletter \@addtoreset{equation}{section} \makeatother

\theoremstyle{plain}
\numberwithin{equation}{section}
\newtheorem{theorem}{Theorem}[section]
\newtheorem{lemma}{Lemma}[section]

\newtheorem{remark}{Remark}[section]	

\newtheorem{corollary}{Corollary}[section]
\newtheorem{example}{Example}[section]
\begin{document}
    \begin{center}
	\textbf{\LARGE{Constructions of entanglement-assisted quantum MDS codes from generalized Reed-Solomon codes}}\footnote {  xiujingzheng99@163.com(X. Zheng);
	liqiwangg@163.com(L. Wang);    
	zhushixin@hfut.edu.cn(S. Zhu)
	}\\
    \end{center}

    \begin{center}
	{Xiujing Zheng \ Liqi Wang \ Shixin Zhu}
    \end{center}

    \begin{center}
	\textit{School of Mathematics, Hefei University of Technology, Hefei, 230009, P. R. China}
    \end{center}

	\begin{abstract}
		
		By generalizing the stabilizer quantum error-correcting codes, entanglement-assisted quantum error-correcting (EAQEC) codes were introduced, 
		which could be derived from any classical linear codes via the relaxation of self-orthogonality conditions with the aid of pre-shared 
		entanglement between the sender and the receiver.
		In this paper, three classes of entanglement-assisted quantum error-correcting maximum-distance-separable (EAQMDS) codes are constructed through  generalized Reed-Solomon codes.
		Under our constructions, the minimum distances of our EAQMDS codes are much larger than those of the known EAQMDS codes of the same lengths that consume the same number of ebits.
		Furthermore, some of the lengths of the EAQMDS codes are not divisors of $q^2-1$, which are completely new and unlike all those known lengths existed before.

		 
	\end{abstract}

	\section{Introduction}
	Over recent decades, the quantum information science has developed very rapidly. Quantum error-correcting (QEC) codes were introduced in order to minimize the decoherence phenomenon over the quantum information channel.
	After Calderbank et al.\cite{CRSS98} gave a connection between classical linear error-correcting codes and QEC codes, the research of QEC codes has made rapidly progress. 
	By utilizing self-orthogonal classical error-correcting codes,  a number of QEC codes with favorable parameters have been derived
	(see \cite{KZL13, LZLX13, KZL14, GLL21, LG09, L04, WZ15} and the relevant references therein). 
	Nevertheless, the self-orthogonality condition forms an obstacle to the construction of QEC codes. 
	Later, a breakthrough had been made by Brun et al. in \cite{BDH06}, where entanglement-assisted (EA) stabilizer formalism was proposed, which utilizes pre-shared entanglement among the sender as well as the receiver to construct QEC codes.
	The associated codes are known as entanglement-assisted quantum error-correcting (EAQEC) codes, that may be constructed from arbitrary classical linear error-correcting codes without the self-orthogonality constraint.
	After that, numerous researchers have been working on constructing EAQEC codes with good parameters via classical linear error-correcting codes (see \cite{SHB11, GHMR19, GJG18, HDB07,GWF23} and the references therein).
	
	Suppose that $q$ is a prime power, an $[[n,k,d;c]]_q$ EAQEC code can correct at most $\lfloor\frac{d-1}{2}\rfloor$ errors by encoding $k$ information qudits into $n$ channel qudits 
	with the assistance of $c$ copies of maximally entangled states.
	In particular, if $c=0$, it is the so-called standard $[[n,k,d]]_q$ QEC code. 
	Analogous to the quantum Singleton bound for the parameters of QEC codes, the more general entanglement-assisted (EA)-quantum Singleton bound for the parameters of EAQEC codes is available to us in the following.

	\begin{theorem}\label{the1.1}
		(\cite{AAHM22,LA18,BDH06,GHW22})(EA-quantum Singleton bound) For any $[[n,k,d;c]]_q$ EAQEC code with $d \le \frac{n+2}{2}$ must satisfy 
		$$n-k+c \ge 2(d-1),$$ 
		where $0\le c \le n-1$. Furthermore, it is indeed the quantum Singleton bound if $c=0$.
	\end{theorem}

	Specifically, an EAQEC code is called an EAQMDS code if it exactly satisfy  this bound.
	Recently, by using constacyclic codes(including cyclic codes and negacyclic codes), generalized Reed-Solomon(GRS) codes as well as extended GRS codes,
	numerous EAQMDS codes have been constructed in a variety of ways.	
	In \cite{FCX16}, due to classical MDS codes of a certain code length, Fan et al. derived several classes of EAQMDS codes with a small number of pre-shared maximally entangled states.
	Subsequently, Lu et al.\cite{LL14} determined the number of maximally entangled states via the decomposition of the defining set of BCH codes and many EAQMDS codes with larger minimum distances were constructed.
	This approach was extended to constacyclic codes by Lu et al. \cite{LLGML18} and Chen et al. \cite{CHFC17}, respectively, and several new classes of EAQMDS codes were derived.
	Thereafter, by using the decomposition of the defining set of constacyclic codes, many classes of EAQMDS codes of lengths dividing $q^2\pm 1$ have been derived
	(see \cite{CZJL21, LMLMLC18, CZK18, LLLM18, PZW21,WZS20, WWZ22, SK19,K19,SK21,GZLF23} and the relevant references therein).
	At the same time, due to the excellent algebraic structure of the GRS codes, a significant number of EAQMDS codes were also constructed from GRS codes.
	In \cite{GJG18}, Guenda et al. developed a relationship between the hull of the classical linear error-correcting codes
	and the number of maximally entangled states which is needed to construct EAQEC codes.
	Additionally, by studying the hull of the GRS codes, they also derived some new EAQMDS codes. 
	Inspired by the ideas of \cite{GJG18},  many classes of MDS codes with hulls of arbitrary dimensions were studied by Luo et al.\cite{LCC19} 
	via GRS codes, and new classes of EAQMDS codes were obtained.
	Besides, a number of other EAQMDS codes were derived via GRS codes as well as extended GRS codes (see \cite{LC19,LZLK19,FFLZ20,TZ20,GLLW20,QZ19}).

	Notably, the lengths of most EAQMDS codes mentioned above are divisors of $q^2\pm 1$. Therefore, scholars are also committed to construct  EAQMDS codes of lengths not dividing $q^2 \pm 1$.
	Guo et al.\cite{GL20} extended the lengths of EAQMDS codes by adding $1$ to the lengths in \cite{GLLW20}, so the lengths may not divide $q^2\pm 1$.
	Jin et al.\cite{JCL21} derived several new classes of EAQMDS codes from GRS codes over finite fields of odd characteristic $q$, 
	whose lengths may not be divisors of $q^2\pm 1$ and can reach $(q+1)(q-3)$.
	Very recently, Wang and Li \cite{WJ22} constructed two classes of EAQMDS codes of the lengths that are sums of two divisors of $q^2-1$ from GRS codes.
	
	Going on the line of the above work, we construct three classes of EAQMDS codes with parameters $[[n, n-2d+c+2, d; c]]_q$ based on GRS codes as follows:

	(1)  $n=\frac{b({q^2}-1)}{a}+\frac{{q^2}-1}{a}$, $a\vert (q+1)$, $a+b \equiv 1\ ({\rm mod}\ 2)$, $2\leq d\leq \frac{a+b+1}{2}\cdot \frac{q+1}{a}$ and $c=b+1$.
	
	(2)  $n=\frac{b({q^2}-1)}{a}+\frac{{q^2}-1}{a}$, $a\vert (q+1)$, $a+b \equiv 0\ ({\rm mod}\ 2)$, $2\leq d\leq \frac{a+b+2}{2}\cdot \frac{q+1}{a}-1$ and $c=b+1$.
	
	(3)  $n=\frac{b({q^2}-1)}{a}$, $a\vert (q-1)$, $(a,b)\neq (q-1,q-1)$, $2\leq d\leq  \frac{b(q-1)}{a}+1$ and $c=b$.

	Compared to standard QEC codes, the EAQEC codes enhance communication capabilities at the expense of pre-shared entanglement.
    In this paper, it is worth noting that some of the lengths of EAQMDS codes are the sum of two divisors of $q^2-1$,
	which implies that the lengths of our codes might not be divisors of $q^2-1$. Some of the lengths are new and have never been covered by the lengths available in the
    literature. This extends the length of the EAQMDS codes. Also, compared with known EAQMDS codes, our codes have larger minimum distances, which 
    enhance the  error-correction capability.
	The paper is structured as follows. Some basic concepts and results about GRS codes and EAQEC codes are reviewed in Section \ref{sec2}.
	In Section \ref{sec3}, three new classes of EAQMDS codes are obtained from GRS codes. 
	Section \ref{sec4} gives a conclusion.

	\section{Preliminaries}\label{sec2}
	
	Let $\mathbb{F}_{q^2}$ be the finite field with $q^2$ elements, where $q$ is a prime power.  
	A $q^2$-ary linear code $\mathcal{C}$ is denoted by $[n,k,d]_{q^2}$ if its length is $n$, its dimension is $k$ 
	and its minimum distance is $d$, which is a linear subspace of $\mathbb{F}^n_{q^2}$.
	Moreover, the minimum distance $d$ of the code $\mathcal{C}$ must satisfy the following well-known bound.

	\begin{theorem}\label{th2.1}
		\cite{MS77}(Singleton bound)  An $[n,k,d]$ linear code $\mathcal{C}$ over $\mathbb{F}_{q^2}$ must satisfy
		$$n-k\geq d-1.$$
	\end{theorem}
	A linear code is called an maximum-distance-separable (MDS) code if the equality holds. 

	For any two vectors $\textbf{x}=(x_0,x_1,\ldots,x_{n-1})$ and $\textbf{y}=(y_0,y_1,\ldots,y_{n-1})$ belong to $\mathbb{F}_{q^2}^n$, their Euclidean and Hermitian inner product are defined as 
	$$\langle \textbf{x},\textbf{y}\rangle_E=x_0y_0+x_1y_1+\cdots+x_{n-1}y_{n-1},$$
	and 
	$$\langle \textbf{x},\textbf{y}\rangle_H=x_0y_0^q+x_1y_1^q+\cdots+x_{n-1}y_{n-1}^q,$$
	respectively.
	The Euclidean dual code and the Hermitian dual code of a linear code $\mathcal{C} $ with parameters $[n,k,d]_{q^2}$ is respectively defined by
	$$\mathcal{C}^{\bot_E}=\{\textbf{x}\in \mathbb{F}_{q^2}^n\vert \langle \textbf{x},\textbf{y}\rangle_E=0, \forall\  \textbf{y}\in \mathcal{C}\},$$
	$$\mathcal{C}^{\bot_H}=\{\textbf{x}\in \mathbb{F}_{q^2}^n\vert \langle \textbf{x},\textbf{y}\rangle_H=0, \forall\ \textbf{y}\in \mathcal{C}\}.$$
	The code $\mathcal{C}$ is said to be Euclidean (Hermitian) self-orthogonal if $\mathcal{C}\subseteq \mathcal{C}^{\bot_E}$ ( $\mathcal{C}\subseteq \mathcal{C}^{\bot_H}$).
	
	Let $\textbf{a}=(\alpha_1, \alpha_2, \ldots, \alpha_n) \in \mathbb{F}_{q^2}^n$ and $\textbf{v}=(v_1, v_2, \ldots, v_n) \in (\mathbb{F}^*_{q^2})^n$,
	where $\alpha_1, \alpha_2, \ldots, \alpha_n$ are $n$ distinct elements of $\mathbb{F}_{q^2}$ and
	$v_1, v_2, \ldots, v_n$ are $n$ nonzero elements of  $\mathbb{F}_{q^2}$ ($v_i$ can be the same).
	For each integer $k$ with $1\leq k\leq n$, let 
	$$\mathbb{F}_{q^2}^n[x]_k=\{f(x)\in \mathbb{F}_{q^2}^n[x] \vert deg(f(x))\leq k-1 \}.$$
	Then the GRS code of length $n$ and dimension $k$ is defined as:
	$$GRS_k(\textbf{a},\textbf{v})=\{v_1f(\alpha _1), v_2f(\alpha _2), \ldots, v_nf(\alpha _n) \vert for \ all \ f(x)\in \mathbb{F}_{q^2}^n[x]_k\}.$$
	As everyone knows, the GRS code is exactly an MDS code with parameters $[n,k,n-k+1]$ over $\mathbb{F}_{q^2}$ 
	and it has a generator matrix shown below.

	$$ G_k  =  \begin{pmatrix}
		v_1              & v_2                   & \cdots & v_n             \\
		v_1\alpha _1     & v_2\alpha_2           & \cdots & v_n\alpha_n                 \\
		\vdots           & \vdots                & \ddots & \vdots                 \\
		v_1\alpha _1^{k-1} & v_2\alpha_2 ^{k-1}   & \cdots & v_n\alpha_n^{k-1} \\
	
	\end{pmatrix}.$$
	
	The following famous result illustrates the dual code of a GRS code.

	\begin{theorem}\cite{MS77}\label{the2.2}
		The dual code of $GRS_k(\textbf{a},\textbf{v})$ is $GRS_{n-k}(\textbf{a},\textbf{v}^{'})$ for a vector $\textbf{v}^{'}=\{v_1^{'},v_2^{'},\ldots,v_n^{'}\}$,
		such that $v_i^{'}\neq 0$ for any $1\leq i\leq n$.
	\end{theorem}	
	
	\begin{remark}
		According to Theorem \ref{the2.2}, the dual code of an $[n,k,n-k+1]_{q^2}$ $GRS_k(\textbf{a},\textbf{v})$ is still an $[n,n-k,k+1]_{q^2}$ GRS code.
		Additionally, the parity-check matrix of $GRS_{n-k}(\textbf{a},\textbf{v}^{'})$ is the generator matrix of $GRS_k(\textbf{a},\textbf{v})$.
	
	\end{remark}

	For any $\eta \in \mathbb{F}_{q^2}$, the conjugate of $\eta$ is defined as $\overline{\eta}=\eta^q $.
	Let $A=(a_{ij})_{m \times n}$ be an $m \times n$ matrix over $\mathbb{F}_{q^2}$, and define the conjugate of $A$ by $\overline{A}=(\overline{a_{ij}})_{m \times n}$.
	Suppose that $A^\dagger$ is the conjugate transpose of $A$, where $A^\dagger = (\overline{A})^T$ and $A^T$ denotes the transpose of $A$.

	Finally, we give a method to construct EAQEC codes using classical linear error-correcting codes under the Hermitian case.
	\begin{theorem}\cite{BDH06,WB08}\label{the2.3}
		Let $\mathcal{C}$ be an $[n,k,d]_{q^2}$ classical linear error-correcting code, whose  parity-check matrix is $H$ and $H^\dagger $ is the conjugate transpose of $H$, 
		then one can get an $[[n,2k-n+c,d,c]]_q$ EAQEC code, where $c=rank(HH^\dagger )$.
	\end{theorem}

	\section{New EAQMDS codes from GRS codes}\label{sec3}
	In this section, assume that $q$ is a prime power. We will derive three classes of $q$-ary EAQMDS codes by utilizing GRS codes over $\mathbb{F}_{q^2}$.
	Let $\xi$ be a fixed primitive element of $\mathbb{F}_{q^2}$. Suppose that $t=\frac{q^2-1}{a}$, where $a\vert (q+1)$ or  $a \vert (q-1)$.
	Let $\beta =\xi^a$, then $\text{ord}(\beta)=t$.

    \subsection{New EAQMDS codes of length $n=\frac{b({q^2}-1)}{a}+\frac{{q^2}-1}{a}$ with $a\vert (q+1)$}
	In this subsection, we will obtain new $q$-ary EAQMDS codes of length $n=\frac{b({q^2}-1)}{a}+\frac{{q^2}-1}{a}$ with $a\vert (q+1)$ 
	from GRS codes. In the following, we will consider GRS codes of length $n$ by splitting $a+b$ into two cases, i.e., $a+b\equiv 1\ ({\rm mod}\ 2)$ and $a+b\equiv 0\ ({\rm mod}\ 2)$.

	\subsubsection{The case $a+b\equiv 1\ ({\rm mod}\ 2)$}
	To begin with, we give some significant lemmas that will take on essential roles in our constructions.

	\begin{lemma}\label{le3.1}
		Suppose that $a+b\equiv 1\ ({\rm mod}\ 2)$ and $m=\frac{a-b+1}{2}$. If $b\leq min\{a-3,q-3\}$, then there exists a vector $\bm{\rho}=(\rho_0,\rho_1,\ldots, \rho_{b})\in (\mathbb{F}^*_q)^{b+1}$
		such that the following sums 
		$$\sum_{l=0}^{b}\rho_l,\ \sum_{l = 0}^{b}\xi^{mlt}\rho_l,\ \sum_{l = 0}^{b}\xi^{(m+1)lt}\rho_l,\ \ldots,\ \sum_{l=0}^{b} \xi^{(m+b-1)lt}\rho_l$$
		are all non-zeros.

	\end{lemma}
	
	\begin{proof}
	We have to proof that there exists a vector $\bm{\varphi}=(\varphi _0,\varphi _1,\ldots,\varphi _{b})\in  (\mathbb{F}^*_q)^{b+1} $
	such that the group of equations below has a non-zero solution  $\bm{\rho} =(\rho _0,\rho _1,\ldots, \rho _{b})\in (\mathbb{F}^*_q)^{b+1}$.
	\begin{equation}
	    \left\{ 
		\begin{array}{l}
		\sum_{l=0}^{b}\rho_l=\varphi _0,\\
		\\
		\sum_{l=0}^{b}\xi^{mlt}\rho_l=\varphi_1,\\
		\\
		\sum_{l=0}^{b}\xi^{(m+1)lt}\rho_l=\varphi_2,\\
		\\
		\ldots\\
		\\
		\sum_{l = 0}^{b}\xi^{(m+b-1)lt}\rho_l=\varphi_{b}.
	    \end{array}
	    \right.
    \end{equation}
	
	Let $\zeta=\xi^t$.  Denote 
	$$ A=\begin{pmatrix}
		1      & 1              & \cdots & 1 \\
		1      & \zeta^m       & \cdots & \zeta^{bm} \\
		1      & \zeta^{m+1}   & \cdots & \zeta^{b(m+1)}\\
		\vdots & \vdots         & \ddots & \vdots \\
		1      & \zeta^{m+b-1} & \cdots & \zeta^{b(m+b-1)}
	\end{pmatrix}.$$
	The system of Eq.(3.1) could be characterized by the matrix form
	$$A\bm{\rho}^T=(\varphi_0,\varphi_1,\ldots,\varphi_{b})^T=\bm{\varphi}^T.$$
	It can be easily derived that $\text{det}(A)\neq 0$ due to the fact that $A$ is a Vandermonde matrix and 
	$1, \zeta ^m, \zeta ^{m+1}$, $\ldots$, $\zeta ^{m+b-1}$ are all different because of $\zeta $ is a primitive $a$-th root of unity.
	Hence, for a fixed vector $\bm{\varphi}$, the system of Eq.(3.1) has a sole solution $\bm{\rho}=(\rho_0,\rho_1,\ldots, \rho_{b})\in (\mathbb{F}_q)^{b+1}$.
	Next, we will proof that the system of Eq.(3.1) has a non-zero solution.

	Let $\phi_0=(1,1,\ldots,1)$, $\phi_i=(1,\xi^{t(m+i-1)},\ldots, \xi^{tb(m+i-1)})$, for $i=1,\ldots,b$. 
	Then the system of Eq.(3.1) is changed to

	\begin{equation}
		\left\{	 
		\begin{array}{c}
			\phi_0\bm{\rho}^T =\varphi_0, \\
			\phi_1\bm{\rho}^T =\varphi_1,\\
			\phi_2\bm{\rho}^T =\varphi_2,\\
			\vdots \\

			\phi_{b}\bm{\rho}^T =\varphi_{b}.

		\end{array}
		\right.
	\end{equation}

	It can be trivially verified that $\phi^q_i=\phi_{b+1-i}$, for $i=1,\ldots, b$. 
	Hence, 	by taking the $q$-th powers of all the equations in (3.2), one can get 

	\begin{equation}
		\left\{	 
		\begin{array}{c}
			\phi_0\bm{\rho}^T =\varphi_0, \\
			\phi_{b}\bm{\rho}^T =\varphi_1,\\
			\phi_{b-1}\bm{\rho}^T =\varphi_2,\\
			\vdots \\

			\phi_{2}\bm{\rho}^T =\varphi_{b-1},\\
			\phi_{1}\bm{\rho}^T =\varphi_{b}.
		\end{array}
		\right.
	\end{equation}

	Let

	$$\lambda=\left\{ \begin{array}{l}
		\xi^{q+1}        , q\ is\  a\  power\  of\  2,\\
		\xi^\frac{q+1}{2}, q\  is\  an\  odd\  prime\  power .
	\end{array} \right.$$
	
    It is easy to proof that $\lambda^q=-\lambda$. We now  divide $b$ into two cases as follows:
	
	(1) If $b$ is odd, then it can be easily proved that the system of Eq.(3.2) is equivalent to the system of Eq.(3.4) according to (3.3).

	\begin{equation}
		\left\{	 
		\begin{array}{c}
			\phi_0\bm{\rho}^T=\varphi_0, \\
			\\
			(\phi_1+\phi_{b})\bm{\rho}^T=2\varphi_1,\\
			\\
			\vdots \\
			\\
			(\phi_{\frac{b-1}{2}}+\phi_{\frac{b+3}{2}})\bm{\rho}^T=2\varphi_{\frac{b-1}{2}},\\
			\\
			\phi_{\frac{b+1}{2}}\bm{\rho}^T=\varphi_{\frac{b+1}{2}},\\
			\\
			\frac{\phi_1-\phi_{b}}{\lambda }\bm{\rho}^T=0,\\
			\\
			\vdots \\
			\frac{\phi_{\frac{b-1}{2}}-\phi_{\frac{b+3}{2}}}{\lambda}\bm{\rho}^T =0.
		\end{array}
		\right.
	\end{equation}

    Denote
	$$B=\begin{pmatrix}
		\phi_0 \\
		\phi_1+\phi_{b} \\
		\vdots \\
		\phi_{\frac{b-1}{2}}+\phi_{\frac{b+3}{2}} \\
		\phi_{\frac{b+1}{2}}\\
		\frac{\phi_1-\phi_{b}}{\lambda } \\ 
		\vdots \\
		\frac{\phi_{\frac{b-1}{2}}-\phi_{\frac{b+3}{2}}}{\lambda}

	\end{pmatrix}.$$
	
	Obviously, each element of the matrix $B$ belongs to $\mathbb{F}_q$,
	then the system of Eq.(3.4) is represented as $B\bm{\rho}^T=(\varphi _0, 2\varphi _1, \ldots, 2\varphi_{\frac{b-1}{2}},\varphi_{\frac{b+1}{2}}, 0, \ldots, 0)^T$.
	Suppose that
	$$U=\{\bm{\rho}\in\mathbb{F}^b_q \vert B\bm{\rho}^T=(\varphi_0, 2\varphi_1, \ldots, 2\varphi_{\frac{b-1}{2}},\varphi _{\frac{b+1}{2}}, 0, \ldots, 0)^T\}.$$
	For any given non-zero elements $\varphi_0,\varphi_1, \ldots, \varphi _{\frac{b+1}{2}},$ the system of Eq.(3.4) in $U$ has a sole solution.

	When $\varphi_0,\varphi_1,\ldots,\varphi_{\frac{b+1}{2}}$ take all nonzero elements over $\mathbb{F}_q$ respectively, then one obtains
	$$\vert U \vert=(q-1)^{\frac{b+3}{2}}.$$
	Let 
	$$U_i=\{\bm{\rho}\in \mathbb{F}^b_q \vert\rho_i=0,B\bm{\rho}^T=(\varphi_0, 2\varphi_1, \ldots, 2\varphi_{\frac{b-1}{2}}, \varphi_{\frac{b+1}{2}}, 0, \ldots, 0)^T \},$$ 
	where $i=0, 1,\ldots,b$,
	and 
	$$S=\{\bm{\rho}\in (\mathbb{F}^*_q)^b \vert B\bm{\rho}^T =(\varphi_0, 2\varphi_1, \ldots, 2\varphi_{\frac{b-1}{2}}, \varphi_{\frac{b+1}{2}}, 0, \ldots, 0)^T\}. $$
	If $\rho_i=0$, then $\varphi_0, 2\varphi_1, \ldots, 2\varphi_{\frac{b-1}{2}}, \varphi_{\frac{b+1}{2}} $ must meet a linear equation with nonzero coefficients. So 
	$$\vert U_i \vert=(q-1)^{\frac{b+1}{2}}.$$ 
	Accordingly, $S=U\setminus (U_0\cup U_1\cup \cdots \cup U_{b})$. Hence, 
	$\vert S\vert= \vert U \vert-\sum_{i=0}^{b} \vert U_i \vert +\chi$ due to the inclusion-exclusion principle, where $\chi \geq 0$. So
	$$\vert S\vert=(q-1)^{\frac{b+3}{2}}-(b+1)(q-1)^{\frac{b+1}{2}}+\chi.$$
	It is obvious that $\vert S\vert>0$ because of $b\leq q-3$.
	
	(2) If $b$ is even, then the system of Eq.(3.2) is an equivalence of the following system of Eq.(3.5). 
	The remainder of the proof is quite comparable to the case $b$ is odd and we omit it here.

	\begin{equation}
		\left\{	 
		\begin{array}{c}
			\phi_0\bm{\rho}^T =\varphi_0, \\
			\\
			(\phi_1+\phi_{b})\bm{\rho}^T=2\varphi_1,\\
			\\
			\vdots \\
			\\
			(\phi_{\frac{b}{2}}+\phi_{\frac{b+2}{2}})\bm{\rho}^T =2\varphi_{\frac{b}{2}},\\
			\\
			\frac{\phi_1-\phi_{b}}{\lambda }\bm{\rho}^T =0,\\
			\\
			\vdots \\
			\frac{\phi_{\frac{b}{2}}-\phi_{\frac{b+2}{2}}}{\lambda }\bm{\rho}^T =0.
		\end{array}
		\right.
	\end{equation}

	Therefore, there exists a vector $\bm{\rho}=(\rho_0,\rho_1,\ldots, \rho _{b})\in (\mathbb{F}^*_q)^{b+1}$ such that 
	$$\sum_{l = 0}^{b} \rho_l,\ \sum_{l = 0}^{b} \xi ^{mlt}\rho_l,\ \sum_{l = 0}^{b} \xi ^{(m+1)lt}\rho_l,\ \ldots,\ \sum_{l = 0}^{b} \xi^{(m+b-1)lt}\rho_l$$
	are all non-zeros.

    \end{proof}

    \begin{lemma}\label{le3.2}
	Let $a+b \equiv 1 \ ({\rm mod}\ 2)$ and $b\leq a-3$. If $0\leq i,j\leq \frac{a+b+1}{2}\cdot\frac{q+1}{a}-2$ with $(i,j)\neq (0,0)$,
	then $t \vert (qi+j)$ if and only if $qi+j=\upsilon  t$ with $\frac{a-b+1}{2}\leq \upsilon  \leq \frac{a+b-1}{2}$.
    \end{lemma}

    \begin{proof}
	As $0\leq i,j\leq \frac{a+b+1}{2}\cdot\frac{q+1}{a}-2$ with $(i,j)\neq (0,0)$, then $0\leq i,j \leq q-1-\frac{q+1}{a}$ and $0<qi+j\leq q^2-1-\frac{(q+1)^2}{a}$.
	If $t \vert (qi+j)$, then an integer $\upsilon $ exists such that $qi+j=\upsilon  t$ with $0<\upsilon <a-1$. Also, we have  
	$$qi+j=q[\frac{\upsilon (q+1)}{a}-1]+[q-\frac{\upsilon (q+1)}{a}].$$
	Then 
	$$i=\frac{\upsilon (q+1)}{a}-1,\ j=q-\frac{\upsilon (q+1)}{a}.$$
	Since  $0\leq i,j\leq \frac{a+b+1}{2}\cdot \frac{q+1}{a}-2$, one can get
	$$\frac{a-b-1}{2}+\frac{a}{q+1}\leq \upsilon \leq\frac{a+b+1}{2}-\frac{a}{q+1},$$
	which implies that $\frac{a-b+1}{2}\leq \upsilon \leq\frac{a+b-1}{2}$.
	Therefore, $t\vert (qi+j)$ if and only if $qi+j=\upsilon  t$ with $\frac{a-b+1}{2}\leq \upsilon  \leq \frac{a+b-1}{2}$.
    \end{proof}

    \begin{lemma}\label{le3.3}
	
	Let $n=\frac{b(q^2-1)}{a}+\frac{q^2-1}{a}$, and $b\leq min\{a-3,q-3\}$. Assume that $\bm{\tau} =(1,\beta,\beta ^2,\ldots,\beta^{t-1})\in \mathbb{F}_{q^2}^t$ and $\textbf{a}=(\bm\tau, \xi\bm\tau, \xi ^2\bm\tau, \ldots, \xi ^{b}\bm\tau)\in \mathbb{F}_{q^2}^n$.
	Let 
	$$\textbf{v}=(v_0, v_0, \ldots, v_0, \ldots, v_{b}, v_{b}, \ldots, v_{b})_{(1\times (b+1)t)}\in (\mathbb{F}^*_{q^2})^n,$$
	where $v_l^{q+1}=\rho _l$ with $0\leq l\leq {b}$, 
	and $\bm{\rho} =(\rho _0,\rho _1,\ldots, \rho _{b})\in (\mathbb{F}^*_q)^{b+1}$ is a vector satisfy Lemma \ref{le3.1}.
	Then $\langle \textbf{a}^{qi+j}, \textbf{v}^{q+1} \rangle_E \neq 0$ if and only if $(i,j)=(0,0)$ or $qi+j=\upsilon t$ with $\frac{a-b+1}{2}\leq \upsilon \leq \frac{a+b-1}{2}$, 
	where $0\leq i,j\leq \frac{a+b+1}{2}\cdot \frac{q+1}{a}-2$.

    \end{lemma}

    \begin{proof}
	If $(i,j)=(0,0)$, it follows from Lemma \ref{le3.1} that
	$$\langle \textbf{a}^{qi+j}, \textbf{v}^{q+1} \rangle_E=\langle \textbf{a}^0, \textbf{v}^{q+1} \rangle_E=t(v^{q+1}_0+\cdots+v^{q+1}_{b})=t(\rho _0+\cdots+\rho _{b})\neq 0.$$
	If $(i,j)\neq (0,0)$, then
	$$\langle \textbf{a}^{qi+j}, \textbf{v}^{q+1} \rangle_E=\sum_{l = 0}^{b}\xi ^{(qi+j)l}v_l^{q+1}\sum_{s=0}^{t-1}\beta ^{s(qi+j)}.$$
	Notably,           

	$$\sum_{s=0}^{t-1}\beta ^{s(qi+j)}=\left\{ \begin{array}{l}
		0,  t \nmid (qi+j) ,\\
		t,   t \mid (qi+j).
	\end{array} \right.$$
    Based on Lemma \ref{le3.2}, $\sum_{s=0}^{t-1}\beta ^{s(qi+j)}=t$ if and only if $qi+j=\upsilon t$ with $\frac{a-b+1}{2}\leq \upsilon  \leq \frac{a+b-1}{2}$, then by Lemma \ref{le3.1}, one can get 
    $\langle \textbf{a}^{qi+j}, \textbf{v}^{q+1} \rangle_E=t\sum_{l = 0}^{b}\xi ^{\upsilon tl}v_l^{q+1}=t\sum_{l = 0}^{b}\xi ^{\upsilon tl}\rho_l \neq 0.$ The result holds.

    \end{proof}

\begin{theorem}\label{the3.1}
Let $n=\frac{b(q^2-1)}{a}+\frac{q^2-1}{a}$, where q is a prime power, $b\leq min\{a-3,q-3\}$, $a\vert (q+1)$ and $a+b\equiv 1\ ({\rm mod}\ 2) $,
then one can get an EAQMDS code with parameters $[[n, n-2d+c+2, d; c]]_q$, where $2\leq d\leq \frac{a+b+1}{2}\cdot \frac{q+1}{a}$ and $c=b+1$.

\end{theorem}

\begin{proof}
	Assume that there exists a GRS code, denoted as $GRS_k(\textbf{a},\textbf{v})$,  associated with vectors \textbf{a} and \textbf{v}, where \textbf{a} and \textbf{v} are given in Lemma \ref{le3.3}. 
	One can get $GRS_k(\textbf{a},\textbf{v})$ has a generator matrix as below.
	$$ G_k  =  \begin{pmatrix}
		v_0    & v_0             & \cdots & v_0                     & \cdots & v_{b}                      & \cdots & v_{b}      \\
		v_0    & v_0\beta       & \cdots & v_0\beta ^{t-1}         & \cdots & v_{b}\xi ^{b}            & \cdots & v_{b}(\xi^{b}\beta ^{t-1})         \\
		v_0    & v_0\beta ^2     & \cdots & v_0(\beta ^{t-1})^2     & \cdots & v_{b}(\xi^{b})^2        & \cdots & v_{b}(\xi^{b}\beta ^{t-1})^2             \\
		\ldots & \ldots          & \ldots & \ldots                  & \ldots & \ldots                   & \ldots & \ldots        \\
		v_0    & v_0\beta ^{k-1} & \cdots & v_0(\beta ^{t-1})^{k-1} & \cdots & v_{b}(\xi^{b})^{k-1}    & \cdots & v_{b}(\xi^{b}\beta ^{t-1})^{k-1}
	
	\end{pmatrix}.$$
	According to Theorem \ref{the2.2}, there exists a $GRS_{n-k}(\textbf{a},\textbf{v}^{'})$ with parameters $[n,n-k,k+1]$, whose parity-check matrix is $G_k$. By calculation,
	one can get 
	
	$$ G_kG^\dagger _k  =  \begin{pmatrix}
		\sigma_{0,0}    & \sigma_{1,0}       & \cdots & \sigma_{k-1,0}                 \\
		\sigma_{0,1}    & \sigma_{1,1}       & \cdots & \sigma_{k-1,1}                \\
		\sigma_{0,2}    & \sigma_{1,2}       & \cdots & \sigma_{k-1,2}                  \\
		\ldots          & \ldots             & \ldots &  \ldots                       \\
		\sigma_{0,k-1}  & \sigma_{1,k-1}     & \cdots & \sigma_{k-1,k-1}

	\end{pmatrix},$$
	where $\sigma_{i,j}=\langle \textbf{a}^{qi+j}, \textbf{v}^{q+1} \rangle_E$.
	
	Based on Lemma \ref{le3.3}, $\sigma_{i,j}\neq 0$ if and only if $(i,j)=(0,0)$ or $qi+j=\upsilon t$ with $\frac{a-b+1}{2}\leq \upsilon \leq \frac{a+b-1}{2}$.
	
	If there exist $i_1=i_2=i$ such that $qi+j_1=\upsilon_1 t$ and $qi+j_2=\upsilon_2 t$, where $j_1\neq j_2$, then $j_1-j_2=(\upsilon_1-\upsilon_2)t$. 
	In fact, $\vert j_1-j_2\vert < q-1$. However, $\vert (\upsilon_1-\upsilon_2)t\vert =\vert (\upsilon_1-\upsilon_2)\frac{q+1}{a}(q-1) \vert \geq q-1$.
	So $\sigma_{i,j}\neq 0$ cannot appear in the same row of the matrix.
	
	If there exist $j_1=j_2=j$ such that $qi_1+j=\upsilon_1 t$ and $qi_2+j=\upsilon_2 t$, where $i_1\neq i_2$, then $q(i_1-i_2)=(\upsilon_1-\upsilon_2)t=(\upsilon_1-\upsilon_2)\frac{q^2-1}{a}$.
	Therefore, $q \vert (\upsilon_1-\upsilon_2) $, which contradicts to the fact that $\vert \upsilon_1-\upsilon_2\vert \leq b-1 <q $. So $\sigma_{i,j}\neq 0$ cannot appear in the same column of the matrix.
	
	Hence, for $0\leq i,j\leq \frac{a+b+1}{2}\cdot \frac{q+1}{a}-2$,  $\sigma_{i,j}\neq 0$ cannot occur in the same row and column of the matrix.
	Consequently, $rank(G_kG_k^\dagger )=b+1$. According to Theorem \ref{the2.3}, the EAQMDS codes are derived.

\end{proof}

\begin{remark}
	Taking $b=0$ in Theorem 3.1, then $a$ must be odd. Let $a=2\ell+1$ ($\ell\geq 1$),
	then one can get EAQMDS codes with the following parameters: 
	\begin{itemize}

		\item $[[\frac{q^2-1}{2\ell+1},\frac{q^2-1}{2\ell+1}-2d+3,d;1]]_q$, where $2\leq d\leq \frac{\ell+1}{2\ell+1}(q+1)$.
		
	\end{itemize}

	Actually, EAQMDS codes with the same length had also been constructed in \cite{FCX16} and \cite{LZLK19}. Taking $n=\frac{q^2-1}{2\ell+1}$ in Theorem 3 of \cite{FCX16} and $c=1$ in Theorem 3.7 of \cite{LZLK19},
	we can derive the following two subclasses of EAQMDS codes:
	\begin{itemize}

    \item $[[\frac{q^2-1}{2\ell+1},\frac{q^2-1}{2\ell+1}-2d+3,d;1]]_q,$ where $2\leq d \leq 2(\frac{q+1}{2\ell+1}-1)$.
    
	\item $[[\frac{q^2-1}{2\ell+1},\frac{q^2-1}{2\ell+1}-2d+3,d;1]]_q,$ where $\frac{q+1}{2\ell+1}+2\leq d \leq \frac{\ell+1}{2\ell+1}(q+1)$.
    \end{itemize}

	By comparing our results with the same length as in \cite{FCX16}, due to
	$\frac{t+1}{2t+1}(q+1)>2(\frac{q+1}{2t+1}-1)$,
	one can see that the minimum distance of our EAQMDS code is larger than theirs.
	While comparing our results with \cite{LZLK19}, one can easily see that our codes have the same largest minimum distance as theirs within such case, 
	but the minimum distance $d$ has a wider value range. 

\end{remark}

\begin{remark}
	Let $a\equiv 2\ ({\rm mod}\ 4) $ and $b=\frac{a}{2}$, then $a+b\equiv 1\ ({\rm mod}\ 2)$.
	According to Theorem 3.1, EAQMDS codes can be obtained with the following parameters.
	\begin{itemize}

		\item $[[\frac{q^2-1}{2}+\frac{q^2-1}{a},\frac{q^2-1}{2}+\frac{q^2-1}{a}-2d+c+2,d;c]]_q$, where $2\leq d \leq \frac{3(q+1)}{4}+\frac{q+1}{2a}$, $c=\frac{a}{2}+1$.
		
	\end{itemize}

	In fact, EAQMDS codes of length $n=\frac{q^2-1}{2}+\frac{q^2-1}{a}$ had already been studied in \cite{WJ22}, but the value of $c$ in \cite{WJ22} is incorrect.
	It should be $\frac{a}{2}+1$$($$\frac{b}{2}+1$ for their length$)$.
	Besides, it can be easily seen that our EAQMDS codes coincide with theirs.
	Therefore, the result of ours is a generalization of theirs.
\end{remark}

\begin{remark}
	Let $a\equiv 0\ ({\rm mod}\ 4) $ and $b=\frac{a}{2}+1$, then $a+b\equiv 1\ ({\rm mod}\ 2)$.
    Due to Theorem 3.1, one can get EAQMDS codes with parameters as follows: 
	\begin{itemize}

		\item $[[\frac{q^2-1}{2}+\frac{2(q^2-1)}{a},\frac{q^2-1}{2}+\frac{2(q^2-1)}{a}-2d+c+2,d;c]]_q$, where $2\leq d \leq \frac{3(q+1)}{4}+\frac{q+1}{a}$, $c=\frac{a}{2}+2$.
		
	\end{itemize}

EAQMDS codes of length $n=\frac{q^2-1}{2}+\frac{2(q^2-1)}{a}$ had been studied in \cite{WJ22}, but the $c$ in \cite{WJ22} is different from ours.
\end{remark}

\begin{example}
	We show some of the new EAQMDS codes of length $n = \frac{b({q^2}-1)}{a}+\frac{{q^2}-1}{a}$ with $a+b\equiv 1\ ({\rm mod}\ 2)$ obtained from Theorem \ref{the3.1} whose lengths are not  divisors of $q^2- 1$ in Table \ref{tab1}, 
	\begin{table}
		\caption{New EAQMDS codes of length $n=\frac{b({q^2}-1)}{a}+\frac{{q^2}-1}{a}$ with $a+b\equiv 1\ ({\rm mod}\ 2) $ }
		\begin{center}
			\begin{tabular}{ccccc}
				\toprule
				$ q $    &$a$& $b$ & $[[n,k,d;c]]_q$& $d$\\
				\midrule

				$8$ &$9$  & $4$   &$[[35,42-2d,d;5]]_{8}$  & $2\le d\le 7$  \\
				$9$ &$5$  & $2$   &$[[48,53-2d,d;3]]_{9}$ & $2\le d\le 8$ \\
				    &$10$ & $3$   &$[[32,38-2d,d;4]]_{9}$ & $2\le d\le 7$ \\
				    &$10$ & $5$   &$[[48,56-2d,d;6]]_{9}$  & $2\le d\le 8$  \\
				$11$ &$6$  & $3$   &$[[80,86-2d,d;4]]_{11}$ & $2\le d\le 10$ \\
				     &$12$  & $7$   &$[[80,90-2d,d;8]]_{11}$ & $2\le d\le 10$ \\
				$13$&$7$ & $2$   &$[[72,77-2d,d;3]]_{13}$ & $2\le d\le 10$ \\
				    & $7$ & $4$   &$[[120,127-2d,d;5]]_{13}$ & $2\le d\le 12$ \\
					&$14$ & $3$   &$[[48,54-2d,d;4]]_{13}$ & $2\le d\le 9$ \\
					&$14$ & $5$   &$[[72,80-2d,d;6]]_{13}$ & $2\le d\le 10$ \\
					&$14$ & $7$   &$[[96,106-2d,d;8]]_{13}$ & $2\le d\le 11$ \\
					&$14$ & $9$   &$[[120,132-2d,d;10]]_{13}$ & $2\le d\le 12$ \\
				$16$  &$17$& $2$   &$[[35,40-2d,d;3]]_{16}$ & $2\le d\le 10$ \\
				      & $17$ &$4$ &$[[75,82-2d,d;5]]_{16}$ & $2\le d\le 11$ \\
				      & $17$& $6$  &$[[105,114-2d,d;7]]_{16}$ & $2\le d\le 12$ \\
				      &$17$ & $8$   &$[[135,146-2d,d;9]]_{16}$ & $2\le d\le 13$ \\
				      & $17$& $10$ &$[[165,178-2d,d;11]]_{16}$ & $2\le d\le 14$ \\
				      &$17$ & $12$  &$[[195,210-2d,d;13]]_{16}$ & $2\le d\le 15$ \\
				$17$  &$9$ & $4$  &$[[160,167-2d,d;5]]_{17}$ & $2\le d\le 14$ \\
				      &$9$ & $6$  &$[[224,233-2d,d;7]]_{17}$ & $2\le d\le 16$ \\
				      &$18$ & $3$  &$[[64,70-2d,d;4]]_{17}$ & $2\le d\le 11$ \\
					  &$18$ & $7$  &$[[128,138-2d,d;8]]_{17}$ & $2\le d\le 13$ \\
					  &$18$ & $9$  &$[[160,172-2d,d;10]]_{17}$ & $2\le d\le 14$ \\
					  &$18$ & $11$  &$[[192,206-2d,d;12]]_{17}$ & $2\le d\le 15$ \\
					  &$18$ & $13$  &$[[224,240-2d,d;14]]_{17}$ & $2\le d\le 16$ \\
					  
				$19$  &$5$  & $2$   &$[[216,221-2d,d;3]]_{19}$ & $2\le d\le 16$ \\
		              &$10$ & $3$   &$[[144,150-2d,d;4]]_{19}$ & $2\le d\le 14$ \\
				      &$10$ & $5$   &$[[216,224-2d,d;6]]_{19}$ & $2\le d\le 16$ \\
				      &$10$ & $7$   &$[[288,298-2d,d;8]]_{19}$ & $2\le d\le 18$ \\
				      &$20$ & $5$   &$[[108,116-2d,d;6]]_{19}$ & $2\le d\le 13$ \\
					  &$20$ & $7$   &$[[144,154-2d,d;8]]_{19}$ & $2\le d\le 14$ \\
					  &$20$ & $11$  &$[[216,230-2d,d;12]]_{19}$ & $2\le d\le 16$ \\
					  &$20$ & $13$  &$[[252,268-2d,d;14]]_{19}$ & $2\le d\le 17$ \\
					  &$20$ & $15$  &$[[288,306-2d,d;16]]_{19}$ & $2\le d\le 18$ \\
				
				$23$ &$6$ & $3$   &$[[352,358-2d,d;4]]_{23}$ & $2\le d\le20$ \\
				     &$8$ & $5$   &$[[396,404-2d,d;6]]_{23}$ & $2\le d\le 21$ \\
					 &$12$ & $7$   &$[[352,362-2d,d;8]]_{23}$ & $2\le d\le 20$ \\
				     &$12$ & $9$   &$[[440,452-2d,d;10]]_{23}$ & $2\le d\le 22$ \\
				     &$24$ & $9$   &$[[220,232-2d,d;10]]_{23}$ & $2\le d\le 17$ \\
				     &$24$ & $13$   &$[[308,324-2d,d;14]]_{23}$ & $2\le d\le 19$ \\
			         &$24$ & $15$   &$[[352,370-2d,d;16]]_{23}$ & $2\le d\le 20$ \\
				     &$24$ & $17$   &$[[396,416-2d,d;18]]_{23}$ & $2\le d\le 21$ \\
				     &$24$ & $19$  &$[[440,462-2d,d;20]]_{23}$ & $2\le d\le  22$ \\

					  \bottomrule
			\end{tabular}
		\end{center}
		\label{tab1}
	\end{table}

\end{example}

\subsubsection{The case $a+b\equiv 0\ ({\rm mod}\ 2)$}
Likewise, we will first give some useful lemmas that will serve  essential roles in the construction.

	\begin{lemma}\label{le3.4}
		Suppose that $a+b \equiv 0\ ({\rm mod}\ 2)$ and $m=\frac{a-b}{2}$. If $b\leq min\{a-4, q-3\}$, then there exists a vector $\bm{\rho}=(\rho_0,\rho_1,\ldots, \rho_{b})\in (\mathbb{F}^*_q)^{b+1}$
		such that the following sums 
		$$ \sum_{l=0}^{b} \xi^{(mt-q-1)l}\rho_l,\ \sum_{l=0}^{b} \xi^{[(m+1)t-q-1]l}\rho_l,\ \sum_{l=0}^{b} \xi^{[(m+2)t-q-1]l}\rho_l, \ldots,\ \sum_{l=0}^{b} \xi^{[(m+b)t-q-1]l}\rho_l$$
		are all non-zeros.

	\end{lemma}
	\begin{proof}
		We have to proof that there exists a vector $\bm{\varphi}=(\varphi_1,\ldots,\varphi_{b+1})\in   (\mathbb{F}^*_q)^{b+1}$
		such that the group of equations below has a non-zero solution  $\bm{\rho}=(\rho_0,\rho_1,\ldots, \rho_{b})\in (\mathbb{F}^*_q)^{b+1}$.
	\begin{equation}
	    \left\{ 
		\begin{array}{l}
		\sum_{l = 0}^{b} \xi^{(mt-q-1)l}\rho_l=\varphi_1,\\
		\\
		\sum_{l = 0}^{b} \xi^{[(m+1)t-q-1]l}\rho_l=\varphi_2,\\
		\\
		\sum_{l = 0}^{b} \xi^{[(m+2)t-q-1]l}\rho_l=\varphi_3,\\
		\\
		\ldots\\
		\\
		\sum_{l = 0}^{b} \xi^{[(m+b)t-q-1]l}\rho_l=\varphi_{b+1}.
	    \end{array} 
	    \right.
    \end{equation}
	
	Let $\zeta=\xi^{mt-q-1}$. Denote
	$$ A  =  \begin{pmatrix}

		1      & \zeta       & \cdots & \zeta^b \\
		1      & \zeta \xi^{t}   & \cdots & \zeta^b \xi^{bt}\\
		1      & \zeta \xi^{2t}   & \cdots & \zeta^b \xi^{b(2t)}\\
		\vdots & \vdots         & \ddots & \vdots \\
		1      & \zeta \xi^{bt} & \cdots & \zeta^b \xi^{b(bt)}
	\end{pmatrix}.$$
	The system of Eq.(3.6) could be characterized by the matrix form
	$$A\bm{\rho}^T=(\varphi_1,\varphi_2,\ldots,\varphi_{b+1})^T=\bm{\varphi}^T.$$
	It is easy to show that $\text{det}(A)\neq 0$ due to the fact that $A$ is a Vandermonde matrix and $\zeta, \zeta \xi^{t}, \zeta \xi^{2t}, \ldots$, $\zeta \xi^{bt}$ are all different.
	Hence, for a fixed vector $\bm{\varphi}$, the system of Eq.(3.6) has a unique solution $\bm{\rho}=(\rho_0,\rho_1,\ldots, \rho_{b})\in (\mathbb{F}_q)^{b+1}$.
	Next, we will proof that the system of Eq.(3.6) has a non-zero  solution.

	Let $\bm{\phi}_i=(1,\xi^{(m+i-1)t-q-1},\ldots, \xi^{[(m+i-1)t-q-1]b})$, for $i=1,\ldots, b+1$. 
	Then the system of Eq.(3.6) is becoming

	\begin{equation}
		\left\{	 
		\begin{array}{c}
			\phi_1\bm{\rho}^T=\varphi_1,\\
			\phi_2\bm{\rho}^T=\varphi_2,\\
			\phi_3\bm{\rho}^T=\varphi_3,\\
			\vdots \\

			\phi_{b+1}\bm{\rho}^T=\varphi_{b+1}.
		\end{array}
		\right.
	\end{equation}

	It can be trivially verified that $\phi^q_i=\phi_{b+2-i}$, for $i=1,\ldots, b+1$.	
	Hence, 	by taking the $q$-th powers of all the equations in (3.7), one can get 
	\begin{equation}
		\left\{	 
		\begin{array}{c}
			\phi_{b+1}\bm{\rho}^T =\varphi_1,\\
			\phi_{b}\bm{\rho}^T =\varphi_2,\\
			\phi_{b-1}\bm{\rho}^T =\varphi_3,\\
			\vdots \\
			\phi_{2}\bm{\rho}^T =\varphi_{b},\\
			\phi_{1}\bm{\rho}^T =\varphi_{b+1}.
		\end{array}
		\right.
	\end{equation}

	Let

	$$\lambda=\left\{ \begin{array}{l}
		\xi^{q+1}        , q\ is\  a\  power\  of\  2,\\
		\xi^\frac{q+1}{2}, q\  is\  an\  odd\  prime\  power .
	\end{array} \right.$$
	
    It is easy to proof that $\lambda^q=-\lambda$. We now divide $b$ into two cases as follows:
	
	(1) If $b$ is odd, then it is easy to get that the system of Eq.(3.7) is equivalent to the system of Eq.(3.9) according to (3.8).

	\begin{equation}
		\left\{	 
		\begin{array}{c}
			(\phi_1+\phi_{b+1})\bm{\rho}^T =2\varphi_1,\\
			\\
			\vdots \\
			\\
			(\phi_{\frac{b+1}{2}}+\phi_{\frac{b+3}{2}})\bm{\rho}^T =2\varphi_{\frac{b+1}{2}},\\
			\\
			\frac{\phi_1-\phi_{b+1}}{\lambda }\bm{\rho}^T =0,\\
			\\
			\vdots \\
			\frac{\phi_{\frac{b+1}{2}}-\phi_{\frac{b+3}{2}}}{\lambda}\bm{\rho}^T =0.
		\end{array}
		\right.
	\end{equation}

    Denote
	$$B= \begin{pmatrix}
		\phi_1+\phi_{b+1} \\
		\vdots \\
		\phi_{\frac{b+1}{2}}+\phi_{\frac{b+3}{2}}\\
		\frac{\phi_1-\phi_{b+1}}{\lambda} \\ 
		\vdots \\
		\frac{\phi_{\frac{b+1}{2}}-\phi_{\frac{b+3}{2}}}{\lambda}

	\end{pmatrix}.$$
	
	It can be easily derived that each element of the matrix $B$ belongs to $\mathbb{F}_q$,
	then the system of Eq.(3.9) may be represented as $B\bm{\rho}^T=(2\varphi_1, \ldots, 2\varphi_{\frac{b+1}{2}}, 0, \ldots, 0)^T$. 
	Suppose that
	$$U=\{\bm{\rho}\in \mathbb{F}^{b+1}_q \vert B\bm{\rho}^T=(2\varphi_1, \ldots, 2\varphi_{\frac{b+1}{2}}, 0, \ldots, 0)^T\}.$$
	For any provided non-zero elements $\varphi_1, \ldots, \varphi_{\frac{b+1}{2}},$ the system of Eq.(3.9) in $U$ has a sole solution.
	When $\varphi_1, \ldots, \varphi_{\frac{b+1}{2}}$ take all nonzero elements over $\mathbb{F}_q$ respectively, then one obtains
	$$\vert U \vert=(q-1)^{\frac{b+1}{2}}.$$
	Suppose that
	$$U_i=\{\bm{\rho}\in \mathbb{F}^{b+1}_q \vert \rho _i=0,B\bm{\rho}^T =(2\varphi_1, \ldots, 2\varphi_{\frac{b+1}{2}}, 0, \ldots, 0)^T \},$$ where $i=0, 1, \ldots,b$,
	and 
	$$S=\{\bm{\rho}\in (\mathbb{F}^*_q)^{b+1} \vert B\bm{\rho}^T =(2\varphi _1, \ldots, 2\varphi _{\frac{b+1}{2}}, 0, \ldots, 0)^T\}. $$
	If $\rho_i=0$, then $ 2\varphi _1, \ldots, 2\varphi _{\frac{b+1}{2}}$ must meet a linear equation with nonzero coefficients. So
	$$\vert U_i \vert=(q-1)^{\frac{b-1}{2}}.$$ 
	Consequently, $S=U\setminus (U_0\cup U_1\cup \cdots \cup U_{b-1})$. Hence, we have 
	$\vert S\vert= \vert U \vert-\sum_{i = 0}^{b} \vert U_i \vert +\chi  $ by the inclusion-exclusion principle, where $\chi \geq 0$. So
	$$\vert S\vert=(q-1)^{\frac{b+1}{2}}-(b+1)(q-1)^{\frac{b-1}{2}}+\chi.$$
	It is obvious that $\vert S\vert>0$ because of $b\leq q-3$.
	
	(2) If $b$ is even, then the system of Eq.(3.7) is equivalent to the following system of Eq.(3.10), 
	the remainder of the proof is quite comparable to the case  $b$ is odd and we omit it here.

	\begin{equation}
		\left\{	 
		\begin{array}{c}
			\phi_0\bm{\rho}^T=\varphi_0, \\
			\\
			(\phi_1+\phi_{b-1})\bm{\rho}^T=2\varphi_1,\\
			\\
			\vdots \\
			\\
			(\phi_{\frac{b-2}{2}}+\phi_{\frac{b+2}{2}})\bm{\rho}^T=2\varphi_{\frac{b-2}{2}},\\
			\\
			\phi_{\frac{b}{2}}\bm{\rho}^T=\varphi_{\frac{b}{2}}, \\
			\\
			\frac{\phi_1-\phi_{b-1}}{\lambda }\bm{\rho}^T=0,\\
			\\
			\vdots \\
			\frac{\phi_{\frac{b-2}{2}}-\phi_{\frac{b+2}{2}}}{\lambda }\bm{\rho}^T=0.
		\end{array}
		\right.
	\end{equation}
	
	Therefore, there exists a vector $\bm{\rho}=(\rho_0,\rho_1,\ldots, \rho_{b})\in (\mathbb{F}^*_q)^{b+1}$ such that 
	$$ \sum_{l=0}^{b} \xi^{(mt-q-1)l}\rho_l,\ \sum_{l=0}^{b} \xi^{[(m+1)t-q-1]l}\rho_l,\ \sum_{l=0}^{b} \xi^{[(m+2)t-q-1]l}\rho_l, \ldots,\ \sum_{l=0}^{b} \xi^{[(m+b)t-q-1]l}\rho_l$$
	are all non-zeros.

\end{proof}

\begin{lemma}\label{le3.5}
	Let $a+b \equiv 0 \ ({\rm mod}\ 2)$ and $b\leq a-4$. If $0\leq i,j\leq \frac{a+b+2}{2}\cdot \frac{q+1}{a}-3$ with $(i,j)\neq (0,0)$,
	then $t \vert (qi+j+q+1)$ if and only if $qi+j+q+1=\upsilon  t$ with $\frac{a-b}{2}\leq \upsilon  \leq \frac{a+b}{2}$.
\end{lemma}

\begin{proof}
	As  $0\leq i,j\leq \frac{a+b+2}{2}\cdot \frac{q+1}{a}-3$ with $(i,j)\neq (0,0)$, one can obtain that $0\leq i,j\leq q-2-\frac{q+1}{a}$ and $0<qi+j+q+1\leq q^2-1-\frac{(q+1)^2}{a}$.
	If $t \vert (qi+j+q+1)$,  then an integer $\upsilon $ exists such that $qi+j+q+1=\upsilon  t$ with $0<\upsilon <a-1$. In addition, we have  
	$$qi+j+q+1=q[\frac{\upsilon (q+1)}{a}-1]+[q-\frac{\upsilon (q+1)}{a}].$$
	Then 
	$$i=\frac{\upsilon (q+1)}{a}-2,\ j=q-\frac{\upsilon (q+1)}{a}-1.$$
	Due to $0\leq i,j\leq \frac{a+b+2}{2}\cdot \frac{q+1}{a}-3$, we have 
	$$\frac{a-b-2}{2}+\frac{q+1}{a} \leq \upsilon \leq\frac{a+b+2}{2}-\frac{q+1}{a},$$
	which  indicates that $\frac{a-b}{2}\leq \upsilon \leq \frac{a+b}{2}$.
	Therefore, $t\vert (qi+j+q+1)$ if and only if $qi+j+q+1=\upsilon t$ with $\frac{a-b}{2}\leq \upsilon \leq \frac{a+b}{2}$.
\end{proof}

\begin{lemma}\label{le3.6}
	
	Let $n=\frac{b(q^2-1)}{a}+\frac{q^2-1}{a}$, and $b\leq min\{a-4,q-3\}$. Assume that $\bm{\tau}=(1,\beta,\beta ^2,\ldots, \beta ^{t-1})\in \mathbb{F}_{q^2}^t$ and $\textbf{a}=(\bm{\tau}, \xi\bm{\tau}, \xi ^2\bm{\tau}, \ldots, \xi^{b}\bm{\tau})\in \mathbb{F}_{q^2}^n$.
	Let
	$$\textbf{v}=(v_0 , v_0\beta, \ldots, v_0\beta ^{t-1} , \ldots, v_{b}, v_{b}\beta , \ldots, v_{b}\beta ^{t-1})_{(1\times (b+1)t)}\in (\mathbb{F}^*_{q^2})^n,$$
	where $v_l^{q+1}=\rho _l$ with $0\leq l\leq {b}$, 
	and $\bm{\rho}=(\rho_0,\rho_1,\ldots, \rho_{b})\in (\mathbb{F}^*_q)^{b+1}$ is a vector satisfy Lemma \ref{le3.4}.
	Then $\langle \textbf{a}^{qi+j}, \textbf{v}^{q+1} \rangle_E \neq 0$ if and only if $qi+j+q+1=\upsilon t$ with $\frac{a-b}{2}\leq \upsilon \leq \frac{a+b}{2}$, 
	where $0\leq i,j\leq \frac{a+b+2}{2}\cdot \frac{q+1}{a}-3$.
\end{lemma}

\begin{proof}
	If $(i,j)=(0,0)$, then 
	$$\langle \textbf{a}^{qi+j}, \textbf{v}^{q+1} \rangle_E=\langle \textbf{a}^0, \textbf{v}^{q+1} \rangle_E=\sum_{l=0}^{b}v_l^{q+1}\sum_{s=0}^{t-1}\beta ^{s(q+1)}.$$
	Note that $t\nmid (q+1)$, then $\sum_{s=0}^{t-1}\beta ^{s(q+1)}=0$. Therefore, $\langle \textbf{a}^{qi+j}, \textbf{v}^{q+1} \rangle_E=0$.

	If $(i,j)\neq (0,0)$, then
	$$\langle \textbf{a}^{qi+j}, \textbf{v}^{q+1} \rangle_E=\sum_{l = 0}^{b}\xi ^{(qi+j)l}v_l^{q+1}\sum_{s=0}^{t-1}\beta ^{s(qi+j+q+1)}.$$
	Note that           

	$$\sum_{s=0}^{t-1}\beta ^{s(qi+j+q+1)}=\left\{ \begin{array}{l}
		0, t \nmid (qi+j+q+1) ,\\
		t, t \mid (qi+j+q+1).
	\end{array} \right.$$
According to Lemma \ref{le3.5}, $\sum_{s=0}^{t-1}\beta ^{s(qi+j+q+1)}=t$ if and only if $qi+j+q+1=\upsilon t$ with $\frac{a-b}{2}\leq \upsilon  \leq \frac{a+b}{2}$, then by Lemma \ref{le3.4}, one can get 
$\langle \textbf{a}^{qi+j}, \textbf{v}^{q+1} \rangle_E=t\sum_{l = 0}^{b}\xi ^{(\upsilon t-q-1)l}v_l^{q+1}=t\sum_{l=0}^{b}\xi ^{(\upsilon t-q-1)l}\rho_l \neq 0.$

\end{proof}

\begin{theorem}\label{the3.2}
Let $n=\frac{b(q^2-1)}{a}+\frac{q^2-1}{a}$, where q is a prime power, $b\leq min\{a-4,q-3\}$, $a\vert (q+1)$ and $a+b\equiv 0\ ({\rm mod}\ 2)$. Then one can get
an EAQMDS code with parameters $[[n, n-2d+c+2, d; c]]_q$, where $2\leq d\leq \frac{a+b+2}{2}\cdot \frac{q+1}{a}-1$ and $c=b+1$.

\end{theorem}

\begin{proof}
	Suppose that there exists a GRS code, denoted as $GRS_k(\textbf{a},\textbf{v})$,  associated with vectors \textbf{a} and \textbf{v}, where \textbf{a} and \textbf{v} are given in Lemma \ref{le3.6}. 
	One can get $GRS_k(\textbf{a},\textbf{v})$ has a generator matrix as below.

	$$ G_k=\begin{pmatrix}
		v_0    & v_0\beta       & \cdots & v_0\beta^{t-1}         & \cdots & v_{b}                    & \cdots & v_{b}\beta^{t-1}        \\
		v_0    & v_0\beta^2     & \cdots & v_0(\beta^{t-1})^2     & \cdots & v_{b}\xi^{b}             & \cdots & v_{b}\xi^{b}(\beta^{t-1})^2             \\
		\ldots & \ldots          & \ldots & \ldots                & \ldots & \ldots                   & \ldots & \ldots        \\
		v_0    & v_0\beta^{k-1} & \cdots & v_0(\beta^{t-1})^{k-1} & \cdots & v_{b}(\xi^{b})^{k-2}     & \cdots & v_{b}(\xi^{b})^{k-2}(\beta^{t-1})^{k-1}\\
		v_0    & v_0\beta^{k} & \cdots & v_0(\beta^{t-1})^{k} & \cdots & v_{b}(\xi^{b})^{k-1}           & \cdots & v_{b}(\xi^{b})^{k-1}(\beta^{t-1})^{k}
	
	\end{pmatrix}.$$

	According to Theorem \ref{the2.2}, there exists a $GRS_{n-k}(\textbf{a},\textbf{v}^{'})$ with parameters $[n,n-k,k+1]$, whose parity-check matrix is $G_k$. By calculation,
	one can obtain that 
	
	$$ G_kG^\dagger_k=\begin{pmatrix}
		\sigma_{0,0}    & \sigma_{1,0}       & \cdots & \sigma_{k-1,0}                 \\
		\sigma_{0,1}    & \sigma_{1,1}       & \cdots & \sigma_{k-1,1}                \\
		\sigma_{0,2}    & \sigma_{1,2}       & \cdots & \sigma_{k-1,2}                  \\
		\ldots          & \ldots             & \ldots &  \ldots                       \\
		\sigma_{0,k-1}  & \sigma_{1,k-1}     & \cdots & \sigma_{k-1,k-1}
	
	\end{pmatrix}.$$
	where $\sigma_{i,j}=\langle \textbf{a}^{qi+j}, \textbf{v}^{q+1} \rangle_E$.
	
	According to Lemma \ref{le3.6}, $\sigma_{i,j}\neq 0$ if and only if  $qi+j+q+1=\upsilon t$ with $\frac{a-b}{2}\leq \upsilon \leq \frac{a+b}{2}$.
	
	If there exist $i_1=i_2=i$ such that $qi+j_1+q+1=\upsilon_1 t$ and $qi+j_2+q+1=\upsilon_2 t$, where $j_1\neq j_2$, then $j_1-j_2=(\upsilon_1-\upsilon_2)t$. 
	In fact, $\vert j_1-j_2\vert < q-1$. However, $\vert (\upsilon_1-\upsilon_2)t\vert =\vert (\upsilon_1-\upsilon_2)\frac{q+1}{a}(q-1) \vert \geq q-1$.
	It is a contradiction. So $\sigma_{i,j}\neq 0$ cannot appear in the same row of the matrix.
	
	If there exist $j_1=j_2=j$ such that $qi_1+j+q+1=\upsilon_1 t$ and $qi_2+j+q+1=\upsilon_2 t$, where $i_1\neq i_2$, then $q(i_1-i_2)=(\upsilon_1-\upsilon_2)t=(\upsilon_1-\upsilon_2)\frac{q^2-1}{a}$.
	Therefore, $q \vert(\upsilon_1-\upsilon_2)$, which contradicts to the fact that $\vert \upsilon_1-\upsilon_2\vert \leq b <q $. So $\sigma_{i,j}\neq 0$ cannot appear in the same column of the matrix.

	Hence, for $0\leq i,j\leq \frac{a+b+2}{2}\cdot \frac{q+1}{a}-3$,  $\sigma_{i,j}\neq 0$ cannot occur in the same row and column of the matrix.
	Consequently, $rank(G_kG_k^\dagger )=b+1$. According to Theorem \ref{the2.2}, the EAQMDS codes are derived.

\end{proof}

\begin{remark}
	Taking $b=0$ in Theorem \ref{the3.2}, then $a$ must be even. Let $a=2\ell$ ($\ell\geq 2$),
	then one can get EAQMDS codes with the following parameters: 
	\begin{itemize}

		\item $[[\frac{q^2-1}{2\ell},\frac{q^2-1}{2\ell}-2d+3,d;1]]_q$, where $2\leq d\leq \frac{\ell+1}{2\ell}(q+1)-1$.  
		
	\end{itemize}
	In fact, EAQMDS codes with the same length had also been constructed in \cite{FCX16} and \cite{LZLK19}. Taking $n=\frac{q^2-1}{2\ell}$ in Theorem 3 of \cite{FCX16} and $c=1$ in Theorem 3.3 of \cite{LZLK19},
	we can derive two subclasses of EAQMDS codes as follows:
	\begin{itemize}

		\item $[[\frac{q^2-1}{2\ell},\frac{q^2-1}{2\ell}-2d+3,d;1]]_q,$ where $2\leq d \leq \frac{q+1}{\ell}-2$.
		
		\item $[[\frac{q^2-1}{2\ell},\frac{q^2-1}{2\ell}-2d+3,d;1]]_q,$ where $\frac{q+1}{2\ell}+2\leq d \leq  \frac{\ell+1}{2\ell}(q+1)$.
	\end{itemize}

	Comparing our results of the same length as in \cite{FCX16}, one can see that the minimum distance of our EAQMDS codes is larger than theirs
	because of $\frac{q+1}{2} \geq\frac{q+1}{\ell}-2$ when $\ell\geq 2$. Comparing our results with \cite{LZLK19}, 
	it can be easily seen that the largest minimum distance of ours is one less than theirs, but the minimum distance $d$ has a wider value range.

\end{remark}

By taking $a=q+1$ in Theorems \ref{le3.1} and \ref{le3.2}, respectively, the following corollary is clearly established.

\begin{corollary}\label{co3.1}
	Let $q$ be a prime power, and $b^{'}\leq q-2$, then one can get a $[[b^{'}(q-1),b^{'}(q-1)-2d+2+c,d;c]]_q$ EAQMDS code,  where $2\leq d\leq \lfloor \frac{q+b^{'}+1}{2}\rfloor $ and $c=b^{'}$. 
\end{corollary}

\begin{remark}
Quantum MDS codes of the length in Corollary \ref{co3.1} had been constructed in Theorem 5 of \cite{GLL21}. 
Compared with it, the minimum distance of our EAQMDS code is much larger than that of the quantum MDS code of the same length in \cite{GLL21}(see Table \ref{tab3}).
	\begin{table}
		\caption{A comparison between quantum MDS codes and EAQMDS codes}
		\begin{center}
			\begin{tabular}{ccc}
				\toprule
				Parameters& $d$ & References    \\
				\midrule
				$[[b^{'}(q-1),b^{'}(q-1)-2d+2,d]]_q$       & $2\leq d\leq \lfloor \frac{b^{'}q-1}{q+1}\rfloor+1$ &   \cite{GLL21}  \\
				$[[b^{'}(q-1),b^{'}(q-1)-2d+2+c,d;c]]_q$   &$2\leq d\leq \lfloor \frac{q+b^{'}+1}{2}\rfloor$, $c=b^{'}$ &  Ours       \\
				\bottomrule
			\end{tabular}
		\end{center}
		\label{tab2}
	\end{table}
	
\end{remark}

\begin{example}
	We show some of the new EAQMDS codes of length $n=\frac{b({q^2}-1)}{a}+\frac{{q^2}-1}{a}$ with $a+b\equiv 0\ ({\rm mod}\ 2)$ obtained from Theorem \ref{the3.2} whose lengths are not divisors of $q^2-1$ in Table \ref{tab3}.

	\begin{table}
		\caption{New EAQMDS codes of length $n = \frac{b({q^2} - 1)}{a}+\frac{{q^2} - 1}{a}$ with $a+b\equiv 0\ ({\rm mod}\ 2) $ }
		\begin{center}
			\begin{tabular}{ccccc}
				\toprule
				$ q $    &$a$& $b$ & $[[n,k,d;c]]_q$& $d$\\
				\midrule
				$7$ &$8$  & $2$   &$[[18,23-2d,d;3]]_{7}$  & $2\le d\le 5$  \\
				    &$8$  & $4$   &$[[30,37-2d,d;5]]_{7}$  & $2\le d\le 6$  \\
				
				$8$ 
				    &$9$  & $1$   &$[[14,18-2d,d;2]]_{8}$  & $2\le d\le 5$  \\
				    &$9$  & $3$   &$[[28,34-2d,d;4]]_{8}$  & $2\le d\le 6$  \\
				    &$9$  & $5$   &$[[42,50-2d,d;6]]_{8}$  & $2\le d\le 7$  \\
					
				$9$ &$5$ & $1$   &$[[32,36-2d,d;2]]_{9}$  & $2\le d\le 7$  \\
				    &$10$ & $2$   &$[[24,29-2d,d;3]]_{9}$ & $2\le d\le 6$ \\
				    &$10$ & $6$   &$[[56,65-2d,d;7]]_{9}$  & $2\le d\le 8$  \\
					
				$11$
				    &$12$  & $4$   &$[[50,57-2d,d;5]]_{11}$ & $2\le d\le 8$ \\
				    &$12$  & $6$   &$[[70,79-2d,d;7]]_{11}$ & $2\le d\le 9$ \\
				    &$12$  & $8$   &$[[90,101-2d,d;9]]_{11}$ & $2\le d\le 10$ \\		   
				
				$13$&$7$ & $1$   &$[[48,52-2d,d;2]]_{13}$ & $2\le d\le 9$ \\
				    &$7$ & $3$   &$[[96,102-2d,d;4]]_{13}$ & $2\le d\le 11$ \\
				    &$14$ & $2$   &$[[36,41-2d,d;3]]_{13}$ & $2\le d\le 8$ \\
					&$14$ & $4$   &$[[60,67-2d,d;5]]_{13}$ & $2\le d\le 9$ \\
					&$14$ & $8$   &$[[108,119-2d,d;9]]_{13}$ & $2\le d\le 11$ \\
					&$14$ & $10$   &$[[132,145-2d,d;11]]_{13}$ & $2\le d\le 12$ \\
					
				$16$  &$17$ & $1$  &$[[30,34-2d,d;2]]_{16}$ & $2\le d\le 9$ \\
				      & $17$&$3$   &$[[60,66-2d,d;4]]_{16}$ & $2\le d\le 10$ \\
				      & $17$& $5$  &$[[90,98-2d,d;6]]_{16}$ & $2\le d\le 11$ \\
				      &$17$ & $7$  &$[[120,132-2d,d;8]]_{16}$ & $2\le d\le 12$ \\
				      & $17$& $9$  &$[[150,162-2d,d;10]]_{16}$ & $2\le d\le 13$ \\
				      &$17$ & $11$ &$[[180,194-2d,d;12]]_{16}$ & $2\le d\le 14$ \\
					  &$17$ &$13$  &$[[210,226-2d,d;14]]_{16}$ & $2\le d\le 15$ \\

				$17$ &$9$ & $1$  &$[[64,68-2d,d;2]]_{17}$ & $2\le d\le 11$ \\
				     &$9$ & $3$  &$[[128,134-2d,d;4]]_{17}$ & $2\le d\le 13$ \\
				     &$9$ & $5$  &$[[192,200-2d,d;6]]_{17}$ & $2\le d\le 15$ \\
				      &$18$ & $4$  &$[[80,87-2d,d;5]]_{17}$ & $2\le d\le 11$ \\
					  &$18$ & $6$  &$[[112,121-2d,d;7]]_{17}$ & $2\le d\le 12$ \\
					  &$18$ & $10$  &$[[176,189-2d,d;11]]_{17}$ & $2\le d\le 14$ \\
					  &$18$ & $12$  &$[[208,223-2d,d;13]]_{17}$ & $2\le d\le 15$ \\
					  &$18$ & $14$  &$[[240,257-2d,d;15]]_{17}$ & $2\le d\le 16$ \\
					  
				$19$  &$5$ & $1$  &$[[144,148-2d,d;2]]_{19}$ & $2\le d\le 15$ \\
				      &$10$ & $2$  &$[[108,113-2d,d;3]]_{19}$ & $2\le d\le 13$ \\
					  &$10$ & $6$  &$[[252,261-2d,d;7]]_{19}$ & $2\le d\le 17$ \\
			    	  &$20$ & $2$  &$[[54,59-2d,d;3]]_{19}$ & $2\le d\le 11$ \\
					  &$20$ & $6$  &$[[126,135-2d,d;7]]_{19}$ & $2\le d\le 13$ \\
					  &$20$ & $8$  &$[[162,173-2d,d;9]]_{19}$ & $2\le d\le 14$ \\
					  &$20$ & $10$  &$[[198,211-2d,d;11]]_{19}$ & $2\le d\le 15$ \\
					  &$20$ & $12$  &$[[234,249-2d,d;13]]_{19}$ & $2\le d\le 16$ \\
					  &$20$ & $14$  &$[[270,287-2d,d;15]]_{19}$ & $2\le d\le 17$ \\
					  &$20$ & $16$  &$[[306,325-2d,d;17]]_{19}$ & $2\le d\le 18$ \\

					  \bottomrule
			\end{tabular}
		\end{center}
		\label{tab3}
	\end{table}

\end{example}

\subsection{New EAQMDS codes of length $n=\frac{b({q^2}-1)}{a}$ with $a\vert (q-1)$}
In this subsection, some new $q$-ary EAQMDS codes of length $n=\frac{b({q^2}-1)}{a}$ with $a\vert (q-1)$ are obtained from GRS codes. 
We begin with some useful lemmas that will play  major roles in our constructions.

\begin{lemma}\label{le3.7}
	Suppose that $a\vert (q-1)$ and $b\leq a$. If $(a,b)\neq (q-1,q-1)$, then
	there exists a vector $\bm{\rho}=(\rho_0,\rho_1,\ldots, \rho_{b-1})\in (\mathbb{F}^*_q)^{b}$
	such that the following sums 
	$$\sum_{l=0}^{b-1}\rho_l,\ \sum_{l = 0}^{b-1}\xi^{lt}\rho_l,\ \sum_{l = 0}^{b-1}\xi^{2lt}\rho_l,\ \ldots,\ \sum_{l=0}^{b-1} \xi^{(b-1)lt}\rho_l$$
	are all non-zeros.

\end{lemma}

\begin{proof}

We have to proof that there exists a vector $\bm{\varphi}=(\varphi _0,\varphi _1,\ldots,\varphi _{b-1})\in (\mathbb{F}^*_q)^{b}$
such that the group of equations below has a non-zero solution  $\bm{\rho}=(\rho_0,\rho_1,\ldots, \rho_{b-1})\in (\mathbb{F}^*_q)^{b}$.

\begin{equation}
	\left\{ 
	\begin{array}{l}
	\sum_{l=0}^{b-1}\rho_l=\varphi_0,\\
	\\
	\sum_{l=0}^{b-1}\xi^{lt}\rho_l=\varphi_1,\\
	\\
	\sum_{l=0}^{b-1}\xi^{2lt}\rho_l=\varphi_2,\\
	\\
	\ldots\\
	\\
	\sum_{l = 0}^{b-1}\xi^{(b-1)lt}\rho_l=\varphi_{b-1}.
	\end{array}
	\right.
\end{equation}

Let $\zeta=\xi^t$. Denote 
$$ A=\begin{pmatrix}
	1      & 1              & \cdots & 1 \\
	1      & \zeta       & \cdots & \zeta^{b-1} \\
	1      & \zeta^{2}   & \cdots & \zeta^{2(b-1)}\\
	\vdots & \vdots         & \ddots & \vdots \\
	1      & \zeta^{b-1} & \cdots & \zeta^{(b-1)(b-1)}
\end{pmatrix}.$$
The system of Eq.(3.11) could be characterized by the matrix form
$$A\bm{\rho}^T=(\varphi_0,\varphi_1,\ldots,\varphi_{b-1})^T=\bm{\varphi}^T.$$
It is easy to show that $\text{det}(A)\neq 0$ because of $A$ is a Vandermonde matrix and $1, \zeta, \zeta^{2}$, $\ldots$, $\zeta^{b-1}$ are all different. 
Hence, for a fixed vector $\bm{\varphi}$, the system of Eq.(3.11) has a sole solution $\bm{\rho}=(\rho_0,\rho_1,\ldots, \rho_{b-1})\in (\mathbb{F}_q)^{b}$.
Since $(a,b)\neq (q-1,q-1)$ and $\zeta $ is a primitive $a$-th root of unity, then there must exist a $\zeta^{'}\in \mathbb{F}^*_q$ such that $\zeta^{'}\notin\{ 1, \zeta, \zeta^{2}$, $\ldots$, $\zeta^{b-1}\}$.
Therefore, we can take $\varphi_0=1$, $\varphi_1=\zeta^{'}$, $\varphi_2=(\zeta^{'})^2$, $\ldots$, $\varphi_{b-1}=(\zeta^{'})^{b-1}$.
According to the Cramer's rule, we can get 

$$\rho_0=\frac{|A_0|}{|A|}, \rho_1=\frac{|A_1|}{|A|}, \ldots, \rho_{b-1}=\frac{|A_{b-1}|}{|A|},$$
where 
$$|A_0|=
\left|
\begin{array}{ccccc}
1                & 1           & \cdots & 1   \\
\zeta^{'}        & \zeta       &  \cdots & \zeta^{b-1}   \\
(\zeta^{'})^2    & \zeta^{2}   & \cdots&   \zeta^{2(b-1)}   \\
\cdots           & \vdots      & \cdots &   \vdots   \\
(\zeta^{'})^{b-1}& \zeta^{b-1} & \cdots &     \zeta^{(b-1)(b-1)}\\
\end{array}
\right|, \ldots, 
|A_{b-1}|=
\left|
\begin{array}{ccccc}
1      & 1           & \cdots & 1   \\
1      & \zeta       &  \cdots & \zeta^{'}  \\
1      & \zeta^{2}   & \cdots &  (\zeta^{'})^2\\
\cdots & \vdots      & \cdots &   \cdots     \\
1      & \zeta^{b-1} & \cdots &(\zeta^{'})^{b-1}\\
\end{array}
\right|.
$$

Since $A_0,A_1,\ldots,A_{b-1}$ are all Vandermonde matrixes and $1, \zeta ^1, \zeta^{2}$, $\ldots$, $\zeta^{b-1}, \zeta^{'} $ are different over $\mathbb{F}^*_q$.
Hence, $|A_0|,|A_1|,\ldots,|A_{b-1}|$ are all non-zeros, which implies that $\rho_0,\rho_1,\ldots, \rho _{b-1}$ are all non-zeros.

Therefore, there exists a vector $\bm{\rho}=(\rho_0,\rho_1,\ldots, \rho_{b-1})\in (\mathbb{F}^*_q)^{b}$ such that 
$$\sum_{l=0}^{b-1}\rho_l,\ \sum_{l = 0}^{b-1}\xi^{lt}\rho_l,\ \sum_{l = 0}^{b-1}\xi^{2lt}\rho_l,\ \ldots,\ \sum_{l=0}^{b-1} \xi^{(b-1)lt}\rho_l$$
are all non-zeros.

\end{proof}

\begin{lemma}\label{le3.8}
Let $a\vert (q-1)$ and $b\leq a$. If $0\leq i,j\leq \frac{b(q-1)}{a}-1$,
then $t \vert (qi+j)$ if and only if $qi+j=\upsilon  t$ with $0\leq \upsilon  \leq b-1$.
\end{lemma}

\begin{proof}
Since  $0\leq i,j\leq \frac{b(q-1)}{a}-1$, one can obtain that $0\leq i,j \leq q-2$ and $0\leq qi+j< q^2-1$.
If $t \vert (qi+j)$, then an integer $\upsilon $ exists such that $qi+j=\upsilon  t$ with $0\leq \upsilon <a$. Additionally, 
$$qi+j=q[\frac{\upsilon (q-1)}{a}]+[\frac{\upsilon (q-1)}{a}].$$
Then 
$$i=\frac{\upsilon (q-1)}{a},\ j=\frac{\upsilon (q-1)}{a}.$$
Since  $0\leq i,j\leq \frac{b(q-1)}{a}-1$, we have 
$$0\leq \upsilon \leq b-\frac{a}{q-1},$$
which implies that $0\leq \upsilon \leq b-1$.
Therefore, $t\vert (qi+j)$ if and only if $qi+j=\upsilon  t$ with $0\leq \upsilon \leq b-1$.
\end{proof}

\begin{lemma}\label{le3.9}

Let $n=\frac{b(q^2-1)}{a}$, where $a\vert (q-1)$, $b\leq a$ and $(a,b)\neq (q-1,q-1)$.  Assume that $\bm{\tau} =(1,\beta,\beta ^2,\ldots,\beta^{t-1})\in \mathbb{F}_{q^2}^t$ and $\textbf{a}=(\bm\tau, \xi\bm\tau, \xi ^2\bm\tau, \ldots, \xi ^{b-1}\bm\tau)\in \mathbb{F}_{q^2}^n$.
Let 
$$\textbf{v}=(v_0, v_0, \ldots, v_0, \ldots, v_{b-1}, v_{b-1}, \ldots, v_{b-1})_{(1\times bt)}\in (\mathbb{F}^*_{q^2})^n,$$
where $v_l^{q+1}=\rho _l$ with $0\leq l\leq {b-1}$, 
and $\bm{\rho} =(\rho _0,\rho _1,\ldots, \rho _{b-1})\in (\mathbb{F}^*_q)^{b}$ is a vector satisfy Lemma \ref{le3.7}.
Then $\langle \textbf{a}^{qi+j}, \textbf{v}^{q+1} \rangle_E \neq 0$ if and only if $qi+j=\upsilon t$ with $0\leq \upsilon \leq b-1$, 
where  $0\leq i,j\leq \frac{b(q-1)}{a}-1$.

\end{lemma}

\begin{proof}
From the definitions of $\textbf{a}$ and $\textbf{v}$, we can get 
$$\langle \textbf{a}^{qi+j}, \textbf{v}^{q+1} \rangle_E=\sum_{l = 0}^{b-1}\xi ^{(qi+j)l}v_l^{q+1}\sum_{s=0}^{t-1}\beta ^{s(qi+j)}.$$
Notably,         

$$\sum_{s=0}^{t-1}\beta ^{s(qi+j)}=\left\{ \begin{array}{l}
	0, t \nmid (qi+j) ,\\
	t, t \mid (qi+j).
\end{array} \right.$$
From Lemma \ref{le3.8}, $\sum_{s=0}^{t-1}\beta ^{s(qi+j)}=t$ if and only if $qi+j=\upsilon t$ with $0\leq \upsilon  \leq b-1$, then by Lemma \ref{le3.7}, one can get 
$\langle \textbf{a}^{qi+j}, \textbf{v}^{q+1} \rangle_E=t\sum_{l = 0}^{b-1}\xi ^{\upsilon tl}v_l^{q+1}=t\sum_{l = 0}^{b-1}\xi ^{\upsilon tl}\rho_l \neq 0.$ The result holds.

\end{proof}

\begin{theorem}\label{the3.3}
Let $n=\frac{b(q^2-1)}{a}$, where q is a prime power, $b\leq a$, $a\vert (q-1)$ and $(a,b)\neq (q-1,q-1)$, 
then one can get an EAQMDS code with parameters $[[n, n-2d+c+2, d; c]]_q$, where $2\leq d\leq \frac{b(q-1)}{a}+1$ and $c=b$.

\end{theorem}

\begin{proof}
Suppose that there exists a GRS code, denoted as $GRS_k(\textbf{a},\textbf{v})$,  associated with vectors \textbf{a} and \textbf{v}, where \textbf{a} and \textbf{v} are given in Lemma \ref{le3.9}. 
One can get $GRS_k(\textbf{a},\textbf{v})$ has a generator matrix as below.
$$ G_k  =  \begin{pmatrix}
	v_0    & v_0             & \cdots & v_0                     & \cdots & v_{b-1}                      & \cdots & v_{b-1}      \\
	v_0    & v_0\beta       & \cdots & v_0\beta ^{t-1}         & \cdots & v_{b-1}\xi^{b-1}            & \cdots & v_{b-1}(\xi^{b-1}\beta ^{t-1})         \\
	v_0    & v_0\beta ^2     & \cdots & v_0(\beta ^{t-1})^2     & \cdots & v_{b-1}(\xi^{b-1})^2        & \cdots & v_{b-1}(\xi^{b-1}\beta ^{t-1})^2             \\
	\ldots & \ldots          & \ldots & \ldots                  & \ldots & \ldots                   & \ldots & \ldots        \\
	v_0    & v_0\beta ^{k-1} & \cdots & v_0(\beta ^{t-1})^{k-1} & \cdots & v_{b-1}(\xi^{b-1})^{k-1}    & \cdots & v_{b-1}(\xi^{b-1}\beta ^{t-1})^{k-1}

\end{pmatrix}.$$
According to Theorem \ref{the2.2}, there must exist a $GRS_{n-k}(\textbf{a},\textbf{v}^{'})$ with parameters $[n,n-k,k+1]$, whose parity-check matrix is $G_k$. By calculation,
one can get 

$$ G_kG^\dagger _k  =  \begin{pmatrix}
	\sigma_{0,0}    & \sigma_{1,0}       & \cdots & \sigma_{k-1,0}                 \\
	\sigma_{0,1}    & \sigma_{1,1}       & \cdots & \sigma_{k-1,1}                \\
	\sigma_{0,2}    & \sigma_{1,2}       & \cdots & \sigma_{k-1,2}                  \\
	\ldots          & \ldots             & \ldots &  \ldots                       \\
	\sigma_{0,k-1}  & \sigma_{1,k-1}     & \cdots & \sigma_{k-1,k-1}

\end{pmatrix}.$$
where $\sigma_{i,j}=\langle \textbf{a}^{qi+j}, \textbf{v}^{q+1} \rangle_E$.

From Lemma \ref{le3.9}, $\sigma_{i,j}\neq 0$ if and only if $qi+j=\upsilon t$ with $0\leq \upsilon \leq b-1$.

If there exist $i_1=i_2=i$ such that $qi+j_1=\upsilon_1 t$ and $qi+j_2=\upsilon_2 t$, where $j_1\neq j_2$, then $j_1-j_2=(\upsilon_1-\upsilon_2)t$. 
In fact, $\vert j_1-j_2\vert < q-1$. However, $\vert (\upsilon_1-\upsilon_2)t\vert =\vert (\upsilon_1-\upsilon_2)\frac{q-1}{a}(q+1) \vert \geq q+1$.
So $\sigma_{i,j}\neq 0$ cannot appear in the same row of the matrix.

If there exist $j_1=j_2=j$ such that $qi_1+j=\upsilon_1 t$ and $qi_2+j=\upsilon_2 t$, where $i_1\neq i_2$, then $q(i_1-i_2)=(\upsilon_1-\upsilon_2)t=(\upsilon_1-\upsilon_2)\frac{q^2-1}{a}$.
Therefore, $q \vert (\upsilon_1-\upsilon_2) $, which contradicts to the fact that $\vert \upsilon_1-\upsilon_2\vert \leq b-1 <q $. So $\sigma_{i,j}\neq 0$ cannot appear in the same column of the matrix.

Hence, for $0\leq i,j\leq \frac{b(q-1)}{a}-1$,  $\sigma_{i,j}\neq 0$ cannot occur in the same row and column of the matrix.
Consequently, $rank(G_kG_k^\dagger )=b$. According to Theorem \ref{the2.3},
the EAQMDS codes are derived.

\end{proof}

\begin{remark}
	EAQMDS codes with the following parameters had been constructed in \cite{LZLK19}:
		\begin{itemize}
		
		\item $[[b\frac{q^2-1}{2a},b\frac{q^2-1}{2a}-2d+c+2,d;c]]_q$, where $2a|(q+1)$, $2\leq b\leq 2a$, $1\leq c\leq 2a-1$, and $cm+2\leq d\leq (a+\lceil\frac{c}{2}\rceil)m$.
		\item $[[b\frac{q^2-1}{2a+1},b\frac{q^2-1}{2a+1}-2d+c+2,d;c]]_q$, where $(2a+1)|(q+1)$, $2\leq b\leq 2a$, $1\leq c\leq 2a$, and $cm+2\leq d\leq (a+1+\lceil\frac{c}{2}\rceil)m$.
		
	\end{itemize}
	It is easy to see that their code lengths are different from ours due to the fact that our $a$ is a divisor of $q-1$.
\end{remark}

\begin{example}
	We show some of the new EAQMDS codes of length $n=\frac{b({q^2}-1)}{a}$ with $a\vert (q-1)$ derived from Theorem \ref{the3.3} whose lengths are not divisors of $q^2-1$ in Table \ref{tab4}.

	\begin{table}
		\caption{New EAQMDS codes of length $n = \frac{b({q^2} - 1)}{a}$ with $a\vert(q-1)$ }
		\begin{center}
			\begin{tabular}{ccccc}
				\toprule
				$ q $    &$a$& $b$ & $[[n,k,d;c]]_q$& $d$\\
				\midrule
				$4$ &$3$  & $2$   &$[[10,14-2d,d;2]]_{4}$  & $2\le d\le 3$  \\
				    
				$5$ &$4$  & $3$   &$[[30,37-2d,d;3]]_{5}$  & $2\le d\le 6$  \\
				
				$7$ &$6$  & $4$   &$[[32,38-2d,d;4]]_{7}$  & $2\le d\le 5$  \\
				    &$6$  & $5$   &$[[40,47-2d,d;5]]_{7}$  & $2\le d\le 6$  \\

				$8$ &$7$ & $2$   &$[[18,22-2d,d;2]]_{8}$  & $2\le d\le 3$  \\
				    &$7$ & $3$   &$[[27,32-2d,d;3]]_{8}$ & $2\le d\le 4$ \\
				    &$7$ & $4$   &$[[36,42-2d,d;4]]_{8}$  & $2\le d\le 5$  \\	
				    &$7$ & $5$   &$[[45,52-2d,d;5]]_{8}$  & $2\le d\le 6$  \\	
				    &$7$ & $6$   &$[[54,62-2d,d;6]]_{8}$  & $2\le d\le 7$  \\	

				$9$ 
			     	&$4$ & $3$   &$[[60,65-2d,d;3]]_{9}$  & $2\le d\le 7$  \\				
				    &$8$ & $3$   &$[[30,35-2d,d;3]]_{9}$  & $2\le d\le 4$  \\
				    &$8$ & $5$   &$[[50,57-2d,d;5]]_{9}$  & $2\le d\le 6$ \\
				    &$8$ & $6$   &$[[60,68-2d,d;6]]_{9}$  & $2\le d\le 7$  \\
				    &$8$ & $7$   &$[[70,79-2d,d;7]]_{9}$  & $2\le d\le 8$  \\

				$11$
				    &$5$  & $2$   &$[[48,52-2d,d;2]]_{11}$ & $2\le d\le 5$ \\
				    &$5$  & $3$   &$[[72,77-2d,d;3]]_{11}$ & $2\le d\le 7$ \\
					&$5$  & $4$   &$[[96,102-2d,d;4]]_{11}$ & $2\le d\le 9$ \\

				    &$10$  & $3$   &$[[36,41-2d,d;3]]_{11}$ & $2\le d\le 4$ \\
				    &$10$  & $4$   &$[[48,54-2d,d;4]]_{11}$ & $2\le d\le 5$ \\		   
					&$10$  & $6$   &$[[72,80-2d,d;6]]_{11}$ & $2\le d\le 7$ \\	
					&$10$  & $7$   &$[[84,93-2d,d;7]]_{11}$ & $2\le d\le 8$ \\	
					&$10$  & $8$   &$[[96,106-2d,d;8]]_{11}$ & $2\le d\le 9$ \\	
					&$10$  & $9$   &$[[108,119-2d,d;9]]_{11}$ & $2\le d\le 10$ \\	

				$13$&$6$ & $4$   &$[[112,118-2d,d;4]]_{13}$ & $2\le d\le 9$ \\
				    &$6$ & $5$   &$[[140,147-2d,d;5]]_{13}$ & $2\le d\le 11$ \\
				
					&$12$ & $5$   &$[[70,77-2d,d;5]]_{13}$ & $2\le d\le 6$ \\
					&$12$ & $7$   &$[[98,107-2d,d;7]]_{13}$ & $2\le d\le 8$ \\
					&$12$ & $8$   &$[[112,122-2d,d;8]]_{13}$ & $2\le d\le 9$ \\
					&$12$ & $9$   &$[[126,137-2d,d;9]]_{13}$ & $2\le d\le 10$ \\
					&$12$ & $10$   &$[[140,150-2d,d;10]]_{13}$ & $2\le d\le 11$ \\
					&$12$ & $11$   &$[[154,167-2d,d;11]]_{13}$ & $2\le d\le 12$ \\
					
				$16$  &$3$ & $2$  &$[[170,174-2d,d;2]]_{16}$ & $2\le d\le 11$ \\
				      &$5$ & $2$  &$[[102,106-2d,d;2]]_{16}$ & $2\le d\le 7$ \\
				      &$5$ & $3$  &$[[153,158-2d,d;3]]_{16}$ & $2\le d\le 10$ \\
				      &$5$ & $4$  &$[[204,210-2d,d;4]]_{16}$ & $2\le d\le 13$ \\
				      &$15$ & $2$  &$[[34,38-2d,d;2]]_{16}$ & $2\le d\le 3$ \\
				      &$15$ & $3$  &$[[51,56-2d,d;3]]_{16}$ & $2\le d\le 4$ \\
				      &$15$ & $4$  &$[[68,74-2d,d;4]]_{16}$ & $2\le d\le 5$ \\
				      &$15$ & $5$  &$[[85,92-2d,d;5]]_{16}$ & $2\le d\le 6$ \\
				      &$15$ & $6$  &$[[102,110-2d,d;6]]_{16}$ & $2\le d\le 7$ \\
				      &$15$ & $7$  &$[[119,128-2d,d;7]]_{16}$ & $2\le d\le 8$ \\
					  &$15$ & $8$  &$[[136,146-2d,d;8]]_{16}$ & $2\le d\le 9$ \\
					  &$15$ & $9$  &$[[153,164-2d,d;9]]_{16}$ & $2\le d\le 10$ \\			  
					  &$15$ & $10$  &$[[170,182-2d,d;10]]_{16}$ & $2\le d\le 11$ \\
					  &$15$ & $11$  &$[[187,200-2d,d;11]]_{16}$ & $2\le d\le 12$ \\
					  &$15$ & $12$  &$[[204,228-2d,d;12]]_{16}$ & $2\le d\le 13$ \\
					  &$15$ & $13$  &$[[221,236-2d,d;13]]_{16}$ & $2\le d\le 14$ \\
					  &$15$ & $14$  &$[[238,254-2d,d;14]]_{16}$ & $2\le d\le 15$ \\

					  \bottomrule
			\end{tabular}
		\end{center}
		\label{tab4}
	\end{table}

\end{example}

\section{Conclusion}\label{sec4}
	Let $n=\frac{b({q^2}-1)}{a}+\frac{{q^2}-1}{a}$ and $n=\frac{b({q^2}-1)}{a}$. Three classes of EAQMDS codes of length $n$ were derived from GRS codes in this paper.
	Taking different values of $a$ and $b$, some lengths of ours results are divisors of $q^2-1$. Compared with the known results, they have much larger minimum distances.
	Furthermore, as the lengths of our EAQMDS codes in this paper can be viewed as the sum of two divisors of $q^2-1$, so they are probably not  divisors of $q^2-1$.
	Some known EAQMDS codes of lengths not be the divisors of $q^2- 1$ are listed in Table \ref{tab5}. Compared with them, our lengths are new and not covered by them.
	
	\begin{sidewaystable}\centering
	\newcommand{\tabincell}[2]{\begin{tabular}{@{}#1@{}}#2\end{tabular}}
      \centering
		\caption{Some known EAQMDS codes of lengths not divide $q^2-1$}
		\renewcommand{\arraystretch}{1.5}
		\begin{tabular}{lll}
			\hline\noalign{\smallskip}
			Parameters & Constraints &  References  \\
			\noalign{\smallskip}\hline\noalign{\smallskip}
			$[[n=1+r\frac{q^2-1}{h}, n-2k+c,k+1;c]]_q$ & $h \vert (q-1)$, $1\leq r\leq h, (c-1)\frac{q-1}{h}+1 \leq k\leq c \frac{q-1}{h}$, $c=l\frac{q-1}{h}. $ & \cite{GL20}  \\
			$[[n=1+r\frac{q^2-1}{h}, n-2k+c,k+1;c]]_q$ & \tabincell{c}{$h \vert (q-1)$, $1\leq r\leq h$, $1+l\frac{q-1}{h} \leq c\leq (l+1)\frac{q-1}{h}.$\\
			$\frac{1}{2}(\frac{q-1}{h}+1)c+\lceil \frac{c-l-2}{2}\rceil \frac{q-1}{h}+1 \leq k\leq \frac{1}{2}(\frac{q-1}{h}+1)c+\lceil \frac{c-l}{2}\rceil \frac{q-1}{h}.$} & \cite{GL20}  \\
			$[[n=1+r\frac{q^2-1}{2h}, n-2k+c,k+1;c]]_q$ & $\frac{q-1}{h}$ even, $1\leq r\leq 2h$, $(c-1)\frac{q-1}{2h}+1\leq  k\leq c\frac{q-1}{2h}$, $c=l\frac{q-1}{h}.$ & \cite{GL20} \\
			$[[n=1+r\frac{q^2-1}{2h}, n-2k+c,k+1;c]]_q$ & $\frac{q-1}{h}$ odd, $1\leq r\leq 2h$, $(c-1)\frac{q-1}{h}+1\leq  k\leq c\frac{q-1}{h}$, $c=l\frac{q-1}{h}.$ & \cite{GL20} \\
			$[[n=1+r\frac{q^2-1}{2h}, n-2k+c,k+1;c]]_q$ & $2h \vert (q-1)$, $1\leq r\leq 2h$, $(c-1)\frac{q-1}{2h}+1\leq  k\leq c\frac{q-1}{2h}$, $c=l\frac{q-1}{2h}.$ & \cite{GL20} \\
			$[[n=1+r\frac{q^2-1}{2h}, n-2k+c,k+1;c]]_q$ &\tabincell{c}{ $2h \vert (q-1)$, $1\leq r\leq 2h$, $1+l\frac{q-1}{2h}\leq c<(l+1)\frac{q-1}{2h}$,\\
			$\frac{1}{2}(\frac{q-1}{2h}+1)c+\lceil \frac{c-l-2}{2}\rceil\frac{q-1}{2h}+1 \leq  k\leq \frac{1}{2}(\frac{q-1}{2h}+1)c+\lceil \frac{c-l}{2} \rceil\frac{q-1}{2h}$.} & \cite{GL20} \\
			$[[n=lh+mr, n-2d+c, d+1;c]]_q$ &\tabincell{c}{ $s\vert (q+1)$, $t\vert (q-1)$ with $s,t$ even, $l=\frac{q^2-1}{s}$, $m=\frac{q^2-1}{t}$, $1\leq h\leq \frac{s}{2}$, $2\leq r\leq \frac{t}{2}$, \\$1\leq d \leq min\{\frac{s+h}{2}\cdot\frac{q+1}{s}-2,\frac{q+1}{2}+\frac{q-1}{t}-1\}$, $c=h-1$.} & \cite{JCL21} \\  
			$[[n=1+(2e+1)\frac{q^2-1}{2s+1}, n-2k+c, k+1;c]]_q$ & $(2s+1)\vert (q+1), 0\leq e \leq s-1,1\leq k\leq (s+1+e)\frac{q+1}{2s+1}-1, c=2e+1.$ & \cite{JCL21} \\
			$[[n=1+(2e+2)\frac{q^2-1}{2s},   n-2k+c, k+1;c]]_q$ & $2s\vert (q+1), 0\leq e\leq s-2, 1\leq k\leq (s+1+e)\frac{q+1}{2s}-1$, $c=2e+2$.  & \cite{JCL21}\\
			$[[n=1+(2e+1)\frac{q^2-1}{2s},   n-2k+c, k+1;c]]_q$ & $2s\vert (q+1), 0\leq e \leq s-1,1\leq k\leq (s+e)\frac{q+1}{2s}-2 $, $c=2e+1$.       &   \cite{JCL21}\\
			$[[n=\frac{q^2-1}{2}+\frac{q^2-1}{b}, n-2d+c+2,d;c]]_q $              & $ q>3$ is odd, $b\vert(q+1)$, $b\equiv2({\rm mod}\ 4)$, $ 2\leq d\leq \frac{3(q+1)}{4}+\frac{q+1}{2b}$, $c=\frac{b}{2}+1$.& \cite{WJ22}  \\
			$[[n=\frac{q^2-1}{2}+\frac{2(q^2-1)}{b}, n-2d+c+2,d;c]]_q $           & $ q>3$ is odd, $b\vert(q+1)$, $b\equiv0({\rm mod}\ 4)$, $ 2\leq d\leq \frac{3(q+1)}{4}+\frac{q+1}{b}-1$, $c=\frac{b}{2}+1$.& \cite{WJ22}  \\
			$[[n=\frac{b({q^2} - 1)}{a}+\frac{{q^2} - 1}{a},n-2d+c+2,d;c]]_q$     & $a\vert (q+1)$, $ b\leq min\{a-3,q-3\}$, $a+b \equiv 1\ ({\rm mod}\ 2)$, $2\leq d\leq \frac{a+b+1}{2}\cdot \frac{q+1}{a}$, $c=b+1$.& Theorem \ref{the3.1}  \\  
			$[[n=\frac{b({q^2} - 1)}{a}+\frac{{q^2} - 1}{a},n-2d+c+2,d;c]]_q$     & $a\vert (q+1)$, $ b\leq min\{a-4,q-3\}$, $a+b \equiv 0\ ({\rm mod}\ 2)$, $2\leq d\leq\frac{a+b+2}{2}\cdot \frac{q+1}{a}-1$,  $c=b+1$.& Theorem \ref{the3.2}  \\
			$[[n=\frac{b({q^2} - 1)}{a},n-2d+c+2,d;c]]_q$                         & $a\vert (q-1)$, $ b\leq a$, $(a,b)\neq (q-1,q-1)$, $2\leq d\leq \frac{b(q-1)}{a}+1$,  $c=b$.& Theorem \ref{the3.3}  \\

			\noalign{\smallskip}\hline
		\end{tabular}
		\label{tab5}
	\end{sidewaystable}

\section*{Data availability}
Data sharing is not applicable to this article as no datasets were generated or analyzed during the current study. 

\section*{Acknowledgement}
	The work was supported by the National Natural Science Foundation of China (12271137, U21A20428, 12171134).

\end{document}